\documentclass[letterpaper, abstracton, DIV11, 12pt, titlepage=false]{scrartcl}

\usepackage[ruled]{algorithm2e}

\SetAlFnt{\small}
\SetAlCapFnt{\small}
\SetAlCapNameFnt{\small}
\SetAlCapHSkip{0pt}
\IncMargin{-\parindent}
\let\chapter\undefined

\usepackage{enumerate}
\usepackage[round]{natbib}
\usepackage[english]{babel}
\usepackage{amsmath}
\usepackage[usenames,dvipsnames]{xcolor}
\usepackage{amssymb}
\usepackage{mathrsfs}
\usepackage{amsthm}
\usepackage{mathabx}
\usepackage{dsfont}
\usepackage{pifont}
\usepackage{bm}
\usepackage{tikz}

\usepackage{mwe}    
\usepackage{subfig}

\usepackage{bm}
\usepackage{pdfsync}
\usepackage{wrapfig}
\usepackage{microtype}
\usepackage{setspace}
\usepackage[colorlinks=true,bookmarks=true, pdftex, citecolor = blue, linkcolor = blue,urlcolor = blue]{hyperref}

\onehalfspace

\DisableLigatures[f]{encoding = *, family = * }

\def\bf{\normalfont\bfseries}

\definecolor{darkgreen}{RGB}{40,150,70}
\newcommand{\cmark}{\textcolor{darkgreen}{\ding{51}}}
\newcommand{\xmark}{\textcolor{red}{\ding{55}}}

\newcommand{\f}{\varphi}
\newcommand{\g}{\psi}
\newcommand{\h}{h^{\beta}}
\newcommand{\hyb}{(1-\beta)\f + \beta \g}
\newcommand{\bmax}{\beta^{\text{max}}}
\newcommand{\URBIr}{\text{URBI($r$)}}
\newcommand{\NMQ}{\text{$(N,M,\bm q)$}}
\newcommand{\BetaMax}{\textsc{BetaMax}}

\changenotsign

\newtheorem{claim}{Claim}

\theoremstyle{plain}
\newtheorem{theorem}{Theorem}
\newtheorem{lemma}{Lemma}
\newtheorem{proposition}{Proposition}
\newtheorem{corollary}{Corollary}

\newtheorem{fact}{Fact}

\theoremstyle{definition}
\newtheorem{definition}{Definition}

\newtheorem{example}{Example}
\newcommand{\ourrep}{}

\theoremstyle{remark}

\newcommand{\ThmConstructionStatement}{	
Given a setting \NMQ, for any hybrid-admissible pair $(\f,\g)$ we have:
\vspace{-0.4em}
\begin{enumerate}[1.]
\setlength{\itemsep}{0pt}
	\item \label{itm:thm:construction:exist} for any $r < 1$ there exists a non-trivial $\beta > 0$ such that $\h$ is $r$-partially strategyproof, 
	\item \label{itm:thm:construction:mon} the mapping $\beta \mapsto \rho_{\NMQ}(\h)$ is monotonic and decreasing.
\end{enumerate}
}
\newcommand{\ThmConstructionProof}{%
To see statement \ref{itm:thm:construction:mon}, fix an agent $i\in N$, a preference profile $(P_i,P_{-i}) \in \mathcal{P}^N$, a misreport $P_i' \in \mathcal{P}$, and a utility function $u_i \in U_{P_i}$. 
If for any $\beta \in [0,1]$, the hybrid $\h$ is manipulable for $i$ in this situation, then 
\begin{equation}
	\left\langle u_i , \h_{i}(P_i,P_{-i})-\h_{i}(P_i',P_{-i})\right\rangle <0.  
\end{equation}
By linearity, we can decompose the left side to 
\begin{eqnarray}
		& & \left\langle u_i , \h_{i}(P_i,P_{-i})-\h_{i}(P_i',P_{-i})\right\rangle \\
		& = & (1-\beta) \left\langle u_i , \f_{i}(P_i,P_{-i})-\f_{i}(P_i',P_{-i})\right\rangle + \beta  \left\langle u_i , \g_{i}(P_i,P_{-i})-\g_{i}(P_i',P_{-i})\right\rangle
.  
\end{eqnarray}
The first part (with factor $(1-\beta)$ must be non-negative by strategyproofness of $\f$. 
Thus, $\left\langle u_i , \g_{i}(P_i,P_{-i})-\g_{i}(P_i',P_{-i})\right\rangle < 0$. 
This implies that for any $\beta \in [0,B]$, agent $i$ in this fixed situation will prefer truthful reporting to misreporting $P_i'$, and for any $\beta \in (B,1]$, it will strictly prefer misreporting $P_i'$. 
Consequently, the set of utility functions, for which the hybrid $\h$ makes truthful reporting a dominant strategy shrinks as $\beta$ increases. 
Therefore, the maximal bound $r$ for which we can guarantee truthful reporting to be a dominant strategy for any agent with utility in \URBIr\ also shrinks. 
This implies that the mapping $\beta \mapsto \rho(\h)$ is monotonic and decreasing.

The proof for statement \ref{itm:thm:construction:exist} is more challenging. 
Consider a strategyproof mechanism $\f$ and a weakly less varying, upper invariant mechanism $\g$, a fixed setting \NMQ, and a fixed bound $r <1$. 
We must find a mixing factor $\beta > 0 $ such that no agent with a utility satisfying \URBIr\ will find a beneficial manipulation to the hybrid $\h$. 

Let $\bm P = (P_i,P_{-i}) \in \mathcal{P}^N$ be a preference profile, $u_i \in U_{P_i}$ a utility function for agent $i$, and let $P_i' \in \mathcal{P}$ be a potential misreport, where
\begin{equation}
	P_i: a_1 \succ \ldots\succ a_m.
\end{equation}
Suppose that $\g$ changes the assignment for $i$ (otherwise the incentive constraint for the hybrid mechanism is trivially satisfied for this preference profile and misreport by strategyproofness of $\f$). 
By Lemma \ref{lem:util_change_structure}, there exists a rank $L \in \{1,\ldots,m-1\}$ such that the gain in expected utility from reporting $P_i'$ instead of $P_i$ under $\g$ is upper-bounded by
\begin{equation}
	\left\langle u_i, \g_i(P_i',P_{-i}) - \g_i(P_i,P_{-i})\right\rangle \leq u_i(a_{L}) - u_i(a_m),
\end{equation}
and the utility gain from reporting $P_i$ truthfully instead of the misreport $P_i'$ under $\f$ is lower-bounded by
\begin{equation}
	\left\langle u_i, \f_i(P_i,P_{-i}) - \f_i(P_i',P_{-i})\right\rangle \geq \varepsilon \cdot \left( u_i(a_{L}) - u_i(a_{L+1})\right),
\end{equation}
where $\varepsilon > 0$ depends only on the setting and the mechanism $\f$.
Thus, the utility gain from reporting $P_i$ truthfully instead of the misreport $P_i'$ under the hybrid $\h$ is lower bounded by
\begin{eqnarray}
	& & \left\langle u_i, \h_i(P_i,P_{-i}) - \h_i(P_i',P_{-i})\right\rangle \\
	& = & (1-\beta) \left\langle u_i, \f_i(P_i,P_{-i}) - \f_i(P_i',P_{-i})\right\rangle + \beta \left\langle u_i, \g_i(P_i,P_{-i}) - \g_i(P_i',P_{-i})\right\rangle \\
	& \geq & (1-\beta)\varepsilon \left(u_i(a_{L}) - u_i(a_{L+1}) \right) - \beta \left( u_i(a_{L}) - u_i(a_m) \right)  \\
	& = & \left( u_i(a_{L}) - u_i(a_{m}) \right) \left( \varepsilon(1-\beta) - \beta \right)
		- \left( u_i(a_{L+1}) - u_i(a_{m}) \right) \left( \varepsilon(1-\beta) \right).
\end{eqnarray}
If $u_i$ satisfies \URBIr, we can lower bound the difference $ u_i(a_{L}) - u_i(a_{m})$ by $r\left( u_i(a_{L+1}) - u_i(a_{m}) \right)$ and get
\begin{eqnarray}
	& & \left\langle u_i, \h_i(P_i,P_{-i}) - \h_i(P_i',P_{-i})\right\rangle \\
	& \geq & \frac{u_i(a_{L+1}) - u_i(a_{m}) }{r} \left( \varepsilon(1-\beta) - \beta - r \varepsilon(1-\beta) \right).
\end{eqnarray}
Since $\frac{u_i(a_{L+1}) - u_i(a_{m}) }{r} \geq 0$, this is positive if and only if
\begin{equation}
\varepsilon(1-\beta) - \beta - r \varepsilon(1-\beta) \geq 0 \;\;\; \Leftrightarrow \;\;\; \beta \leq \frac{\varepsilon (1-r)}{\varepsilon (1-r) + 1}.
\label{eq:beta_bound_for_construction}
\end{equation}
This upper bound for $\beta$ is strictly positive and independent of the specific utility function $u_i$, the preference profile $(P_i, P_{-i})$, and the misreport $P_i'$. 
Therefore, $\h$ is $r$-partially strategyproof if $\beta$ is chosen to satisfy (\ref{eq:beta_bound_for_construction}).
\begin{lemma} 
\label{lem:util_change_structure} 
Consider a setting \NMQ, a strategyproof mechanism $\f$, a weakly less varying, upper invariant mechanism $\g$, an agent $i\in N$, a preference profile $\bm P = (P_i,P_{-i}) \in \mathcal{P}^N$, a misreport $P_i' \in \mathcal{P}$, and a utility function $u_i \in U_{P_i}$. 
If $\f_i(P_i,P_{-i}) \neq \f_i(P_i',P_{-i})$, then there exists $L \in \{1,\ldots,m-1\}$ such that the gain in expected utility from reporting $P_i'$ instead of $P_i$ under $\g$ is upper-bounded by
\begin{equation}
	u_i(a_L) - u_i(a_m),
\end{equation}
and the gain in expected utility from reporting $P_i$ truthfully instead of $P_i'$ under $\f$ is lower-bounded by
\begin{equation}
	\varepsilon \left(u_i(a_L) - u_i(a_{L+1}) \right),
\end{equation}
where $\varepsilon > 0 $ depends only on the setting and the mechanism $\f$.
\end{lemma}
\begin{proof} 
We first introduce the auxiliary concept of the canonical transition. 
Consider two preference orders $P$ and $P'$. 
A \emph{transition} from $P$ to $P'$ is a sequence of preference orders $\tau(P,P') = (P^0, \ldots, P^S)$ such that
\begin{itemize}
\setlength{\itemsep}{0pt}
	\item $P^0 = P$ and $P' = P^S$,
	\item $P^{k+1} \in N_{P^{k}}$ for all $k \in \{0, \ldots, S-1\}$,
\end{itemize}
where $N_P$ is the neighborhood of preference order $P$. 
A transition can be interpreted as a sequence of swaps of adjacent objects that transform one preference order into another if applied in order. 
Suppose,
\begin{equation}
	P': a_1 \succ a_2 \succ \ldots \succ a_m. 
\end{equation} 
Then the \emph{canonical transition} is the transition that results from starting at $P$ and swapping $a_1$ (which may not be in first position for $P$) up until it is in first position. 
Then do the same for $a_2$, until it is in second position, and so on, until $P'$ is obtained.

Suppose $P_i$ corresponds to the preference ordering
\begin{eqnarray*}
P_i & : & a_1 \succ \ldots \succ a_{L-1} \succ a_L \succ \ldots \succ a_m,
\end{eqnarray*}
and let $a_L$ be the best choice object (under $P_i$) for which the assignment under $\f$ changes, i.e.,
\begin{equation}
	\f_{i,a_k}(P_i,P_{-i}) = \f_{i,a_k}(P_i',P_{-i}) \text{ for }k < L, \;\;\; \f_{i,a_L}(P_i,P_{-i}) \neq \f_{i,a_L}(P_i',P_{-i}).
	\label{eq:no_change_trans}
\end{equation}
Consider the canonical transition from $P_i'$ to $P_i$. 
This will bring the objects $a_k, k < L$ into position (as they are under $P_i$) first. 
By Theorem 1 in \citep{MennleSeuken2017PSP_WP} and because $\f$ is strategyproof, the assignment for each of these objects can only weakly increase or weakly decrease. 
However, by (\ref{eq:no_change_trans}) their assignments remain unchanged. 
Therefore, the assignment does not change for any of the swaps that bring the objects $a_k, k < L$ into position. 
Using that $\g$ is weakly less varying than $\f$, we can assume that
\begin{eqnarray*}
P_i' & : & a_1 \succ \ldots \succ a_{L-1} \succ a_L' \succ \ldots \succ a_m'
\end{eqnarray*}
without loss of generality.

By upper invariance of $\g$, the highest gain the agent could obtain from reporting $P_i'$ instead of $P_i$ arises if all probability for its last choice is converted to probability for the best choice for which the assignment can change at all, i.e., $a_L$. 
Thus, the utility gain is bounded by
\begin{equation}
u_i(a_L) - u_i(a_m).
\end{equation}

Let
\begin{equation}
	\varepsilon = \min \left\{ \left| \f_{i,j}(P_i,P_{-i}) - \f_{i,j}(P_i',P_{-i}) \right|~\left|~
	\begin{array}{c}
	j \in M, i \in N, (P_i,P_{-i}) \in \mathcal{P}^N, P_i' \in \mathcal{P}:  \\
	\f_{i,j}(P_i,P_{-i}) \neq \f_{i,j}(P_i',P_{-i})	
	\end{array}
	\right. \right\}
\end{equation}
be the smallest positive amount by which the assignment of some object to some agent can change upon a change of report by that agent under $\f$. 
In the canonical transition from $P_i$ to $P_i'$, the object $a_L$ will only be swapped downwards, i.e., its assignment can not increase in any step. 
But since we assumed that it changes, it must strictly decrease. 
This decrease has at least magnitude $\varepsilon$ by definition. 
Thus, when misreporting, the agent looses at least $\varepsilon$ probability for $a_L$ in some swap. 
From Theorem 1 in \citep{MennleSeuken2017PSP_WP} we know that the assignment for the other object involved in that swap must strictly increase by the same amount $\varepsilon$. 
Since all other swaps reverse the order of objects from ``right'' (as under $P_i$) to ``wrong'' (as under $P_i'$), the assignment can only get weakly worse for the agent. 
Therefore, the gain from reporting $P_i$ truthfully instead of $P_i'$ under $\f$ is at least $\varepsilon \left( u_i(a_{L}) - u_i(a_{L+1}) \right)$.
\end{proof}
This completes the proof of Theorem \ref{thm:construction}
}
\newcommand{\ThmEfficiencyStatement}{	
Given a setting \NMQ, for any mechanisms $\f$ and $\g$, any preference profile $\bm P \in \mathcal{P}^N$, and any mixing factor $\beta \in [0,1]$ the following hold: 
\vspace{-0.4em}
\begin{enumerate}[1.]
\setlength{\itemsep}{0pt}
	\item \label{itm:thm:efficiency:expost} if $\f(\bm P)$ and $\g(\bm P)$ are ex-post efficient at $\bm P$, then $\h(\bm P)$ is ex-post efficient at $\bm P$,
	\item \label{itm:thm:efficiency:dominance} $\g(\bm P)$ ordinally (or rank) dominates $\f(\bm P)$ at $\bm P$ \emph{if and only if} 
		\begin{itemize}
		\setlength{\itemsep}{0pt}
			\item $\h(\bm P)$ ordinally (or rank) dominates $\f(\bm P)$ at $\bm P$, and
			\item $\g(\bm P)$ ordinally (or rank) dominates $\h(\bm P)$ at $\bm P$.
		\end{itemize} 
\end{enumerate}
}
\newcommand{\ThmEfficiencyProof}{%
To see statement \ref{itm:thm:efficiency:expost}, note that $\h(\bm P)$ is a convex combination of (and therefore a lottery over) the assignments $\f(\bm P)$ and $\g(\bm P)$. 
Since both are ex-post efficient, each has a lottery decomposition into ex-post efficient, deterministic assignments. 
Therefore, we can construct a lottery decomposition of $\h(\bm P)$ into ex-post efficient, deterministic assignments by combining the two lotteries. 
This shows ex-post efficiency of $\h(\bm P)$ at $\bm P$.

Suppose, $\g(\bm P)$ ordinally dominates $\f(\bm P)$, i.e., for all $i\in N$ and all $j \in M$ we have 
\begin{equation}
	\sum_{j' \in M:~ P_i: j' \succeq j} \f_{i,j}(\bm P) \leq \sum_{j' \in M:~ P_i: j' \succeq j} \g_{i,j}(\bm P).
\end{equation}
With $\h(\bm P) = (1-\beta) \f(\bm P) + \beta \g(\bm P)$ it follows directly that for any $\beta \in [0,1]$,
\begin{equation}
	\sum_{j' \in M:~ P_i: j' \succeq j} \f_{i,j}(\bm P) \leq \sum_{j' \in M:~ P_i: j' \succeq j} \h_{i,j}(\bm P) \leq \sum_{j' \in M:~ P_i: j' \succeq j} \g_{i,j}(\bm P),
\end{equation}
i.e., $\g$ ordinally dominates $\h$ at $\bm P$, which in turn dominates $\f$. 
Conversely, if $\g(\bm P)$ does not ordinally dominate $\f(\bm P)$, then there exists some agent $i \in N$ and some $j\in M$, such that 
\begin{equation}
	\sum_{j' \in M:~ P_i: j' \succeq j} \f_{i,j}(\bm P) > \sum_{j' \in M:~ P_i: j' \succeq j} \g_{i,j}(\bm P).
\end{equation}
Again, by linearity, this implies 
\begin{equation}
	\sum_{j' \in M:~ P_i: j' \succeq j} \f_{i,j}(\bm P) > \sum_{j' \in M:~ P_i: j' \succeq j} \h_{i,j}(\bm P) > \sum_{j' \in M:~ P_i: j' \succeq j} \g_{i,j}(\bm P),
\end{equation}
which means that $\h$ does not ordinally dominate $\f$ and is not ordinally dominated by $\g$. 
This establishes statement \ref{itm:thm:efficiency:dominance} for ordinal dominance. 
For rank dominance, the result is analogous, where we exploit the fact that the rank distribution of $\h$ is the $\beta$-convex combination of the rank distributions of $\f$ and $\g$. 
}

\newcommand{\PropComputabilityStatement}{	
Given a setting \NMQ, a hybrid-admissible pair of mechanisms $(\f,\g)$, and a bound $\underline\rho \in [0,1]$ Algorithm \ref{alg:compute} (\BetaMax) is complete and correct for the mechanism designer's problem of finding the maximal mixing factor $\bmax_{\NMQ,\f,\g}(\underline\rho)$. 
}
\newcommand{\PropComputabilityProof}{
Since there are only finitely many agents, preference profiles, misreports, and ranks, the loops of the algorithm eventually terminate. 
Thus, the algorithm terminates on any admissible input parameters (i.e., completeness). 

For correctness, we use the fact that by Theorem 4 in \citep{MennleSeuken2017PSP_WP}, $r$-partial strategyproofness is equivalent to strong $r$-partial dominance-strategyproofness. 
Formally, for any agent $i\in N$, any preference profile $(P_i,P_{-i}) \in \mathcal{P}^N$, any misreport $P_i' \in \mathcal{P}$, and any $K \in \{1,\ldots,m\}$, we define the following polynomials (in $r$)
\begin{equation}
	p_K^{\f}(r) = \sum_{j:~ \text{rank}_{P_i}(j) \leq K} r^{\text{rank}_{P_i}(j)} \cdot \left(\f_{i,j}(P_i,P_{-i}) - \f_{i,j}(P_i',P_{-i})\right),
\end{equation} 
\begin{equation}
	p_K^{\g}(r) = \sum_{j:~ \text{rank}_{P_i}(j) \leq K} r^{\text{rank}_{P_i}(j)} \cdot \left(\g_{i,j}(P_i,P_{-i}) - \g_{i,j}(P_i',P_{-i})\right),
\end{equation}
where $\text{rank}_{P_i}(j)$ is the rank of $j$ in the preference order of agent $i$, i.e., the number of objects that $i$ weakly prefers to object $j$. 
For the hybrid mechanism, the correspoinding polynomial is 
\begin{equation}
	p^{\h}_K(r) = (1-\beta)p^{\f}_k(r) + \beta p^{\g}_k(r),
\end{equation}
and $\h$ is $\underline\rho$-partially strategyproof if and only if $p^{\h}_K(\underline\rho) \geq 0$ for all combinations $i, (P_i,P_{-i}), P_i', K$. 
Since $\f$ is strategyproof, $p_K^{\f}(\underline\rho) \geq 0$, and therefore, the only way that $p^{\h}_K(\underline\rho)$ can be negative is for $p^{\g}_K(\underline\rho)$ to be negative. 
Conversely, if $p^{\h}_K(\underline\rho) \geq 0$ for some $\beta$, then $p^{h^{\beta'}}_K(\underline\rho) \geq 0$ for any $\beta' \leq \beta$ as well, i.e., reducing $\beta$ will not lead to a violation of any of the positivity constraints. 
Finally, the only constraints where $\beta$ is not arbitrary are those where $p_K^{\g}(\underline\rho) < 0$ strictly.
In this case, 
\begin{eqnarray}
	& & p_K^{\h}(\underline\rho) = (1-\beta) p_K^{\f}(\underline\rho) + \beta p_K^{\g}(\underline\rho) \geq 0 \\
	& \Leftrightarrow & \beta \leq \frac{p_K^{\f}(\underline\rho)}{p_K^{\f}(\underline\rho)-p_K^{\g}(\underline\rho)}
\end{eqnarray}

Algorithm \BetaMax\ starts with an optimistic guess of $\bmax =1 $ and then reduces this value if this is required to establish a positivity constraint. 
As we observed, subsequent further reductions of $\bmax$ cannot lead to a renewed violation of a previously checked constraint. 
Since the algorithm reduces $\bmax$ only when this is strictly required by some constraint, and this reduction is minimal, the final value of the variable $\bmax$ will be precisely the maximal mixing factor for which $\h$ is $\underline\rho$-partially strategyproof. 	
}

\newcommand{\ThmPSWLVStatement}{	
PS is weakly less varying than RSD.
}
\newcommand{\ThmPSWLVProof}{
Suppose, $n$ agents compete for $m = m_a + 2 + m_b$ objects with capacities given by $\bm q$, and let $M = \{a_1,\ldots,a_{m_a},x,y,b_1, \ldots, b_{m_b}\}$. 
Agent 1 is considering the two preference reports
\begin{eqnarray*}
 P_1 & : & a_1 \succ \ldots \succ a_{m_a} \succ x \succ y \succ b_1 \succ \ldots \succ b_{m_b}, \\
 P_1' & : & a_1 \succ \ldots \succ a_{m_a} \succ y \succ x \succ b_1 \succ \ldots \succ b_{m_b},
\end{eqnarray*}
where the positions of $x$ and $y$ are reversed in the second report. 
The reports of the other agents are fixed and given by $P_{-1}$.

Further suppose that with reports $(P_1,P_{-1})$, the objects where exhausted at times $0< \tau_1 \leq \tau_2 \leq \ldots \leq \tau_m = 1$ under PS. 
Re-label the objects as $j_1, \ldots , j_m$ in increasing order of the times at which they were exhausted. 
If two objects were exhausted at the same time, re-label them in arbitrary order. 
Denote by $\tau_x$ and $\tau_y$ the times at which $x$ and $y$ were exhausted, respectively.

Given these considerations, Claim \ref{claim:PSinvarEquiv} yields equivalent conditions under which PS changes the assignment, Claim \ref{claim:RSDinvarEquiv} yields similar conditions under which RSD changes the assignment, and Claim \ref{claim:PSinvarRSD} shows that the former condition implies the latter. 

\begin{claim}
\label{claim:PSinvarEquiv} 
In Theorem \ref{thm:PS_wlv_RSD}, $\text{PS}_1(P_1,P_{-1}) \neq \text{PS}_1(P_1',P_{-1})$ if and only if
		\begin{enumerate}
\setlength{\itemsep}{0pt}
			\item \label{lemma.PSinvarEquiv.lgoodsoutbeforexy} there exists $k \geq m_a$ such that $\tau_1 \leq \ldots \leq \tau_k < \min(\tau_x, \tau_y) \leq 1$, and
			\item \label{lemma.PSinvarEquiv.asoutbeforexy} for all $l \in \{1,\ldots,m_a\}$ there exists $ l' \in \{1,\ldots, k\}$ with $a_l = j_{l'}$.
		\end{enumerate}
\end{claim}
\begin{proof}
	\begin{description}	
		\item[``$\Rightarrow$''] 
			Choose $k$ such that $j_k$ is the last of the $a_1,\ldots,a_{m_a}$ to run out. 
			Suppose, $\tau_y \leq \tau_k$. 
			Agent 1 is busy consuming shares of other objects until time $\tau_k$, regardless of the reported order of $x$ and $y$. 
			After $\tau_k$ agent 1 consumes shares of $x$ until it is exhausted. 
			Because $y$ was already exhausted before $\tau_k$, agent 1 gets no shares of $y$. 
			Under report $P_1'$, it would finish consuming other objects at $\tau_k$ and find objects $y$ exhausted. 
			Hence, it would begin consuming shares of $x$ immediately, just as it did under report $P_1$. 
			Thus, the order in which $x$ and $y$ are reported does not matter for the times at which it consumes objects $x$ and $y$. 
			Because $P_1$ and $P_1'$ only differ in the order of $x$ and $y$, the remaining objects are also consumed in the same order and at the same times. 
			Hence, agent 1's assignment does not change.
		
			The case for $\tau_x \leq \tau_k$ is analogous.
	
			Because PS is non-bossy \citep{Ekici2012EquilibriumPS}, we know that if the switch from $P_1$ to $P_1'$ did not change the assignment for agent 1, it did not change the assignment at all. 
			
		\item[``$\Leftarrow$''] 
			Suppose the last of the objects $a_1,\ldots,a_{m_a}$ to be exhausted is $j_k$, and $\tau_k < \tau_y \leq \tau_x$. 
			Then agent 1 gets no shares of $y$. 
			If it switches its report to $P_1'$, it will receive a non-trivial share of $y$, hence the assignment changes.
		
			Now suppose the opposite, namely $\tau_y > \tau_x$. 
			Agent 1 begins consumption of $x$ at time $\tau_{k}$ and then turns to $y$ at time $\tau_x$. 
			Thus, agent 1 receives $\tau_x - \tau_{k}$ shares of $x$ and $\tau_y-\tau_x$ shares of $y$. 
			When it switches its report to $P_1'$, it will consume shares of $y$ between $\tau_{k}$ and $\tau_y'$. 
			We need to show that $\tau_y' - \tau_{k} > \tau_y-\tau_x $. 
			If $\tau_y' \geq \tau_y$, this is clear, because $\tau_{k} < \tau_x$ by assumption. 
			In the following we assume $\tau_y' < \tau_y$.
			
			Let $n_y(\tau)$ be the number of agents other than agent 1 consuming shares of $y$ at time $\tau$. 
			$n_y$ is integer-valued and increasing in $\tau$, and there must exist a $\delta > 0$ such that $n_y(\tau_y-\delta) \geq 1$. 
			This means that agent 1 is not the only agent consuming shares of $y$ before it is exhausted. 
			Otherwise, agent 1 would exhaust $y$ alone, which implies that agent 1 received no shares of $x$, a contradiction.
			
			If agent 1 reports $P_1'$ instead, let $n_y'(\tau)$ be the corresponding number of agents consuming $y$ at times $\tau$. 
			We observe that $x$ will be exhausted later, because agent 1 is no longer consuming shares of it. 
			This means that agents who prefer $x$ over $y$ will arrive later at $y$. 
			Agents arriving at $y$ from other objects than $x$ may also arrive later, because they face less competition from the agents stuck at $x$, etc. 
			Therefore $n_y' \leq n_y$.
			
			Under report $P_1$ from agent 1, $y$ is exhausted by $\tau_y$, i.e.,
			\begin{equation}
				q_y = \int_0^{\tau_y} n_y(\tau) + \mathds{1}_{\{\tau \geq \tau_x\}} d\tau, 
				\label{eqn.lemma.PSinvarEquiv.int1} 
			\end{equation}
			and under report $P_1'$, $y$ is exhausted by $\tau_y'$, i.e.,
			\begin{equation}
				q_y = \int_0^{\tau_y'} n_y'(\tau) + \mathds{1}_{\{\tau \geq \tau_{k}\}} d\tau \leq \int_0^{\tau_y'} n_y(\tau) + \mathds{1}_{\{\tau \geq \tau_{k}\}}d\tau.
				\label{eqn.lemma.PSinvarEquiv.int2} 
			\end{equation}
			Equating (\ref{eqn.lemma.PSinvarEquiv.int1}) and (\ref{eqn.lemma.PSinvarEquiv.int2}) gives
			\begin{eqnarray}
				\int_{0}^{\tau_y} n_y(\tau) + \mathds{1}_{\{\tau \geq \tau_x\}} d\tau & \leq & \int_{0}^{\tau_y'} n_y(\tau) + \mathds{1}_{\{\tau \geq \tau_{k}\}} d\tau \\
				\Rightarrow \int_{\tau_y'}^{\tau_y} n_y(\tau) + \mathds{1}_{\{\tau \geq \tau_{k}\}} d\tau & \leq & \int_{0}^{\tau_y'} \mathds{1}_{\{\tau \geq \tau_{k}\}} d\tau - \int_{0}^{\tau_y} \mathds{1}_{\{\tau \geq \tau_{x}\}} dt + \int_{\tau_y'}^{\tau_y} \mathds{1}_{\{\tau \geq \tau_{k}\}} d\tau \\
					& = & \int_{0}^{\tau_y} \mathds{1}_{\{\tau \geq \tau_{k}\}} - \mathds{1}_{\{\tau \geq \tau_{x}\}} d\tau \\
					& = & \tau_x-\tau_{k}.
			\end{eqnarray}
			We know that $j_k$ is exhausted before $\tau_y'$ and hence $n_y(\tau) + \mathds{1}_{\{\tau \geq \tau_{k}\}}\geq 1$ for $\tau \in [\tau_y',\tau_y]$, and $\geq 2$ for $\tau\in [\tau_y-\delta,\tau_y]$. 
			This yields
			\begin{equation}
				\tau_y - \tau_y' < \tau_x - \tau_{k},
			\end{equation}
			or equivalently $\tau_y - \tau_x < \tau_y' - \tau_{k}$.
	\end{description} 
\end{proof}

\begin{claim}
\label{claim:RSDinvarEquiv} 
In Theorem \ref{thm:PS_wlv_RSD}, $\text{RSD}_1(P_1',P_{-1}) \neq \text{RSD}_1(P_1,P_{-1})$ if and only if there exists a sequence $(c_1, \ldots, c_{k_c})$ of $k_c$ agents such that if RSD chose these agents first and in this order, they remove all objects $a_1, \ldots , a_{m_a}$ (and possibly more), but neither $x$, nor $y$.
\end{claim}
\begin{proof} 
	In the RSD mechanism, a permutations of agents is chose amongst all possible permutations with uniform probability. 
	The probability for agent 1 to get some object $j$ is
	\begin{equation}
		P[1 \text{ gets }j] = \frac{|\{\pi \text{ permutation of }N : 1 \text{ gets }j\text{ under }\pi\}|}{|\{\pi \text{ permutation of }N\}|},
	\end{equation}
	where the denominator is $n!$, and each permutation under which agent 1 gets $j$ contributes $\frac{1}{n!}$ to the total probability.
		
	For some permutation $\pi$ consider the turn of agent 1. 
	There are 5 possible cases:
	\begin{enumerate}
\setlength{\itemsep}{0pt}
		\item \label{lemma.RSDinvarEquiv.proof.case1} Agent 1 faces a choice set including some $a_l$'s. 
			This makes no contribution to its chances of getting $x$ or $y$. 
		\item \label{lemma.RSDinvarEquiv.proof.case2} Agent 1 faces a choice set consisting only of $b_l$'s. 
			Again, this makes no contribution to its chances of getting $x$ or $y$.
		\item \label{lemma.RSDinvarEquiv.proof.case3} Agent 1 faces only $b_l$'s and $x$, but not $y$. 
			This case contributes $\frac{1}{n!}$ to its chances of getting $x$. 
			This contribution is independent of the order in which it ranked $x$ and $y$ in its report.
		\item \label{lemma.RSDinvarEquiv.proof.case4} Agent 1 faces only $b_l$'s and $y$, but not $x$. 
			This case contributes $\frac{1}{n!}$ to its chances of getting $y$ and the contribution is again independent of the ranking of $x$ and $y$.
		\item \label{lemma.RSDinvarEquiv.proof.case5} Agent 1 faces $x$, $y$ and some $b_l$'s, but no $a_l$'s. 
			This case contributes $\frac{1}{n!}$ to either the probabilities for $x$ or $y$, depending on the ranking.
	\end{enumerate}
	\begin{description}
	\item[``$\Rightarrow$''] If changing from $P_1$ to $P_1'$ influences the assignment, the assignment for agent 1 must have changed. 
		This is because RSD is non-bossy (by Lemma \ref{lemm.rsd_nonbossy}). 
		RSD also is strategyproof, hence by Theorem 1 in \citet{MennleSeuken2017PSP_WP} the probabilities for objects $x$ and $y$ must have changed. 
		In all but the last case, the chances do not depend on the order in which $x$ and $y$ are reported. 
		Thus, at least one permutation leads to case (\ref{lemma.RSDinvarEquiv.proof.case5}). 
		This means that the sequence of agents chosen prior to agent 1 removes all $a_l$'s, but neither $x$ nor $y$.
	\item[``$\Leftarrow$''] Under report $P_1$, agent 1 will receive $x$ any time case (\ref{lemma.RSDinvarEquiv.proof.case5}) occurs, while under $P_1'$ it will receive $y$. 
		If a sequence $(c_1,\ldots,c_{k_c})$ as defined in Claim \ref{claim:RSDinvarEquiv} exists, it is also the beginning of at least one permutation. 
		When this permutation is selected, case (\ref{lemma.RSDinvarEquiv.proof.case5}) occurs. 
		Switching from report $P_1$ to $P_1'$ thus strictly increases agent 1's chances of getting $y$.
	\end{description} 
	\vspace{-2.3em}
\end{proof}
	
\enlargethispage{.7em}
\begin{claim}
\label{claim:PSinvarRSD} 
In Theorem \ref{thm:PS_wlv_RSD}, \ref{lemma.PSinvarEquiv.lgoodsoutbeforexy}. and \ref{lemma.PSinvarEquiv.asoutbeforexy}. from Claim \ref{claim:PSinvarEquiv} imply the existence of a sequence as described in Claim \ref{claim:RSDinvarEquiv}.
\end{claim}
\begin{proof}
	We prove the claim by constructing a sequence of agents
	\begin{equation}
		(c_1, \ldots,c_{k_c}) = (c_1^{1}, \ldots, c_1^{q_1},\ldots,c_k^{1},\ldots,c_k^{q_k})
	\end{equation} 
	inductively. 
	Under RSD this sequence will remove objects $j_1,\ldots,j_k$ in this order.
	\begin{description}
		\item[Selection of $c_k^1,\ldots,c_k^{q_k}$] By assumption $j_k$ was consumed strictly before $x$, hence $\tau_k < 1$.  
			Then at least $q_k + 1$ agents receive non-trivial shares of $j_k$. 
			Otherwise, if only $q_k$ agents received shares of $j_k$, they would get the entire capacity and take time 1 to consume it, a contradiction. 
			Select $q_k$ of these agents other than agent 1 as $c_k^1,\ldots,c_k^{q_k}$.
			
			Because all $c_k^1,\ldots,c_k^{q_k}$ actually received shares of $j_k$ under PS, they must all prefer $j_k$ to all other objects except for possibly $j_1,\ldots,j_{k-1}$. 
			In other words, suppose that $j_1,\ldots,j_{k-1}$ were removed under RSD in previous turns, the selected agents would remove $j_k$ completely if chosen next (in arbitrary order).
		\item[Selection of $c_l^1,\ldots,c_l^{q_l}, l < k$] Suppose, $c_{l+1}^1,\ldots,c_k^{q_k}$ have been selected. 
			Suppose further that $m_l$ agents (plus possibly agent 1) receive non-trivial shares of $j_l$ under PS. 
			There are two cases:
			\begin{description}
				\item[Case 1] At least $q_l$ of the $m_l$ agents have not been selected as any of the $c_{l+1}^1,\ldots,c_k^{q_k}$ so far. 
					Then these agents are chosen as $c_l^1,\ldots,c_l^{q_k}$.
				\item[Case 2] Only $n_l < q_l$ of the $m_l$ agents have not been selected so far. 
					The rest of the $m_l$ agents have been selected at $k'$ other objects. 
					Let these objects be $j_{\rho(1)},\ldots,j_{\rho(k')}$ with $\rho(l') \in \{l+1,\ldots,k\}$ for all $l' \in \{1,\ldots,k'\}$. 
					At each of the objects $j_{\rho(l')}$, $q_{\rho(l')}$ agents are selected. 
					Now there must be at least $q_l - n_l + 1$ additional agents (possibly including agent 1) consuming non-trivial shares of the objects $j_{\rho(l')}$, otherwise at most $n_l + q_{\rho(1)} + \ldots q_{\rho(k')} + q_l - n_l$ agents fully consume objects $j_l,j_{\rho(1)},\ldots,j_{\rho(k')}$. 
					This will take them until time 1, a contradiction.

					There are two possible cases for these additional $q_l-n_l$ agents (excluding agent 1).
					\begin{description}
						\item[Case 2.1] All of them are available for selection. 
							Then they are selected for the objects $j_{\rho(l')}$ of which they consume non-trivial shares, and the now free agents can be selected for $j_l$.
						\item[Case 2.1] Some of these agents are selected at some other objects $j_{\rho(k'+1)}, \ldots, j_{\rho(k'+k'')}$. 
							Then we use the free agents as in case 2.1, say $n_{l'}$. 
							Then we still need $q_l - n_l - n_{l'}$ agents for $j_l$. There must be at least $q_l + q_{\rho(1)} + \ldots + q_{\rho(k'')} + 1$ agents consuming non-trivial shares of the objects $j_l, j_{\rho(1)}, \ldots, j_{\rho(k'')}$. $q_l - n_l - n_{l'}$ are not selected for any of these objects. 
							Again there are two cases.
					\end{description}
					We repeat this argument inductively until enough agents are found who are still available and can replace agents such that the need at object $j_l$ can be satisfied. 
					This must happen, otherwise all agents selected so far as $c_{l+1}^1,\ldots,c_k^{q_k}$, some $n_l^{'''}< q_l$ agents and possibly agent 1  fully consume objects $j_l, j_{l+1},\ldots,j_k$ objects, again a contradiction.
			\end{description}
	\end{description}
	The fact that all selected agents $c^1_l, \ldots, c^{q_l}_l, l \in \{1,\ldots,k\}$ receive a non-trivial share in the objects $j_l$ implies that they each prefer $j_l$ to all other objects, except possibly $j_1,\ldots,j_{l-1}$. 
	Thus, the sequence $(c_1^1,\ldots,c_k^{q_k})$ has the properties needed for \ref{claim:RSDinvarEquiv}. 
\end{proof}
	
\begin{lemma} 
\label{lemm.rsd_nonbossy} 
For any distribution over orderings, the respective RSD is non-bossy.
\end{lemma}
\begin{proof} 
Fix a distribution over orderings of the agents and let $p_{\pi}$ be the probability that ordering $\pi$ is chosen. 
Suppose that RSD is bossy, then there exists an agents $i,i'$, preference orders $P_i,P_i'$, and $P_{-i} \in \mathcal{P}^{N-i}$ such that $\text{RSD}_i(P_i,P_{-i}) = \text{RSD}_i(P_i',P_{-i})$, but $\text{RSD}_{i'}(P_i,P_{-i}) \neq \text{RSD}_{i'}(P_i',P_{-i})$. 
For the sake of brevity, we write $P$ and $P'$ for $P_i$ and $P_i'$, respectively.

Let $\text{Can}(P,P') = (P_0 = P, P_1, \ldots, P_{k-1}, P_k = P')$ be the canonical transition from $P = P_i$ to $P'=P_i'$. 
As in the proof of Lemma \ref{lem:util_change_structure}, the fact that the assignment is the same at the start and at the end of the transition implies that the assignment never changes during the transition, i.e., $\text{RSD}_i(P_l,P_{-1}) = \text{RSD}_i(P_{l+1},P_{-1})$ for all $l \in \{0,\ldots,k-1\}$. 
Recall that under strategyproof mechanisms, the effect of swaps in the canonical transition is never undone by subsequent swaps and that swaps only effect the probabilities for adjacent objects (see Theorem 1 in \citep{MennleSeuken2017PSP_WP}).
Let $\text{Can}(P,P') = (P_0 = P, P_1, \ldots, P_{k-1}, P_k = P')$ be the canonical transition from $P = P_0$ to $P'=P_k$. As in the proof of Lemma \ref{lem:util_change_structure}, the fact that the assignment is the same at the start and at the end of the transition implies that the assignment never changes during the transition, i.e., $\text{RSD}_i(P_l,P_{-1}) = \text{RSD}_i(P_{l+1},P_{-1})$ for all $l \in \{0,\ldots,k-1\}$. 
Recall that under strategyproof mechanisms, the effect of swaps in the canonical transition is never undone by subsequent swaps and that swaps only effect the probabilities for adjacent objects (see Theorem 1 in \citep{MennleSeuken2017PSP_WP}).

But the assignment changed for agent $i'$, hence it must have changed for agent $i'$ at some swap in the transition, say from $P_{l'}$ to $P_{l'+1} \in N_{P_{l'}}$. Let $j',j''$ be the objects that were swapped in this transition. Consider an ordering of the agents $\pi$ with $p_\pi > 0$. There are two cases.
\begin{itemize}
\setlength{\itemsep}{0pt}
	\item Agent $i$ gets the same object under $P_{l'}$ as under $P_{l'+1}$. Then the swap had no effect on the assignment of any other agent, i.e., under $\pi$ the swap does not change the assignment of the other agents.
	\item Agent $i$ receives $j'$ under $P_{l'}$, but $j''$ under $P_{l'+1}$. Then the swap changes the assignment of the agent that received $j''$ under $P_{l'}$. The magnitude of the change is $-p_{\pi} < 0$. This agent can be $i'$ by assumption.
\end{itemize}
However, the latter case is impossible, because this would also strictly increase agent $i$'s chances of receiving $j''$ (by $p_\pi >0$), implying $\text{RSD}_i(P_{l'},P_{-1}) \neq \text{RSD}_i(P_{l'+1},P_{-1})$, a contradiction. 
\end{proof}
This concludes the proof of Theorem \ref{thm:PS_wlv_RSD}
}

\newcommand{\ThmABMWLVStatement}{	
ABM is weakly less varying than RSD.
}
\newcommand{\ThmABMWLVProof}{
Suppose the following manipulation by agent $i$ by a swap:
\begin{eqnarray*}
	& & P_i: a_1 \succ \ldots \succ a_{m_a} \succ x \succ y \succ b_1 \succ \ldots \succ b_{m_b} \\
	& \rightsquigarrow & P_i' : a_1 \succ \ldots \succ a_{m_a} \succ y \succ x \succ b_1 \succ \ldots \succ b_{m_b}.
\end{eqnarray*}
By Lemma \ref{claim:RSDinvarEquiv}, RSD changes the assignment ($\text{RSD}_i(P_i,P_{-i}) \neq \text{RSD}_i(P_i',P_{-i})$) if and only if there exists an ordering of the agents $\pi$ such that $i$ gets to pick between $x$ and $y$ in its turn, but all objects $a_1,\ldots, a_{m_a}$ are exhausted by higher-ranking agents. 
We show that if ABM changes the assignment, then such an ordering $\pi$ exists. 
Thus, a change of assignment under ABM implies a change of assignment under RSD.

Suppose, the change of report by agent $i$ from $P_i$ to $P_i'$ changes the outcome of ABM for $i$, i.e., $\text{ABM}_i(P_i,P_{-i}) \neq \text{ABM}_i(P_i',P_{-i})$. 
Then from the proof of swap monotonicity \citep{MennleSeuken2017PSP_WP} we know that there exists an ordering of the agents $\pi'$ such that in some round (say $L$), $i$ has not been assigned an object yet, all $a_1,\ldots,a_{m_a}$ are exhausted, but neither $x$ nor $y$ are exhausted. 
Let $r(i')$ be the round in which $i'$ is assigned its object, and let
$$ R(r) = \{i' \in N~|~r(i') = r\} $$
be the set of agents who receive their assignment in round $r$ (given ordering $\pi'$). 
If $i'$ is assigned object $j$ in round $r$, $i'$ has applied to $j$ in that round. 
Thus, out of all the objects with capacity available at the beginning of round $r$, $i'$ must prefer $j$. 
Facing the same set of choices under RSD, $i'$ would also pick $j$.

Consider an ordering $\pi$ that ranks an agent $i'$ before another agent $i''$ if $r(i') < r(i'')$ and ranks them in arbitrary order if $r(i')=r(i'')$. 
Additionally, let $\pi$ rank $i$ after all the agents in the set $R(1) \cup \ldots \cup R(L-1)$. 
If RSD chooses $\pi$ as the ordering of the agents, then all agents in $R(1)$ receive their first choice (as under ABM). 
Next, all agents in $R(2)$ face the choice sets out of which they most prefer the object they were assigned under ABM. 
This continues until finally $i$ faces a choice set that includes none of the $a_1, \ldots, a_{m_a}$, but both $x$ and $y$. 
Hence, $\pi$ is the ordering we are looking for, and its existence concludes the proof. 
}

\newcommand{\ExampleRVnotWLV}{
Consider a setting $N = \{1 ,\ldots, 3\}, M =\{a,b,c\}, q_a=q_b=q_c = 1$. 
For the preference profile
\begin{eqnarray*}
	P_1  & : & a \succ b \succ c, \\
	P_2,P_3 & : & c \succ a \succ b,
\end{eqnarray*}
suppose the rank valuation is $v(1) = 10, v(2) = 6, v(3) = 0$. 
Then RV will assign $b$ to agent 1 with certainty. 
To see this suppose that agent 1 gets $a$ instead. 
Then some other agent $i$ received $b$. 
If agent 1 and agent $i$ trade, the objective increases by $6-10 + 6-0 = 2$. 
Now suppose that agent 1 gets $c$. 
Again some agent $i$ gets object $a$. 
If agent 1 and agent $i$ trade, this improves the objective by $10-0 + 6 - 10 = 6$. 
We have argued that agent 1 will get $b$ in any deterministic assignment chosen by RV with rank valuation $v$. 
Then by definition,
agent 1 must get $b$ with certainty.

Suppose now that agent 1 reports
\begin{eqnarray*}
	P_1'  & : & a \succ c \succ b
\end{eqnarray*}
instead, i.e., it swaps objects $b$ and $c$ in its report. 
Then under any rank efficient assignment (with respect to $(P_1', P_{-1})$), agent 1 will receive object $a$. 
This is because whenever agent 1 gets another object in some deterministic assignment, the objective improves if agent 1 trades with the agent who received $a$ (independent of $v$). 
Since no rank efficient assignment will give agent 1 any other object than $a$, swapping $b$ and $c$ in its report is a beneficial manipulation for agent 1. 
This is independent of its actual utility, as long as the utility is consistent with $P_1$.

Now consider the outcome of RSD: 
it is easy to see that for any ordering of the agents, if agent 1 does not receive $a$ when it gets to choose, object $c$ will not be available. 
Therefore, $\text{RSD}_1(P_1,P_{-1}) = \text{RSD}_1(P_1',P_{-1})$, i.e., RSD does not change the assignment of agent 1. 
This means that RV with the specific choice of rank valuation $v$ is not weakly less varying than RSD, and agent 1 in the given situation would want to manipulate any non-trivial hybrid of RSD and RV.}
\newcommand{\ExampleNBMnotWLV}{
Consider a setting with agents $N = \{1,\ldots,6\}$, objects $M =\{a,\ldots,f\}$, each available in unit capacity. 
Let the agents have preferences
\begin{eqnarray*}
	P_1, P_2 & : & a \succ b \succ c \succ d \succ e \succ f \\
	P_3, P_4, P_5, P_6 & : & c \succ b \succ f \succ d \succ a \succ e.
\end{eqnarray*}
Under RSD, agent $1$'s assignment is $\left(\frac{1}{2},\frac{1}{10}, 0, \frac{7}{30}, \frac{1}{6}, 0 \right)$ of objects $a$ through $f$, respectively. 
Swapping $c$ and $d$ in its report will not change its assignment under RSD. 
Under NBM and truthful reporting, the assignment is the same as under RSD. 
But if agent 1 changes its report by swapping $c$ and $d$, its assignment under NBM changes to $\left(\frac{1}{2},\frac{1}{10}, 0, \frac{2}{5}, 0, 0 \right)$. 
It strictly prefers this assignment in a first order-stochastic dominance sense. 
}
\begin{document}

{\setstretch{1}
\title{%
\LARGE{%
Hybrid Mechanisms: 
Trading Off Strategyproofness and Efficiency of \\ 
Random Assignment Mechanisms}\thanks{\scriptsize{
Department of Informatics, University of Zurich, Switzerland,
Email: \{mennle, seuken\}@ifi.uzh.ch.
We would like to thank (in alphabetical order)
Atila Abdulkadiro\u{g}lu, 
Ivan Balbuzanov, 
Craig Boutilier, 
Gabriel Carroll, 
Lars Ehlers, 
Katharina Huesmann, 
Flip Klijn, 
Antonio Miralles, 
David Parkes, 
and Utku \"Unver
for helpful comments on this work. 
Furthermore, we are thankful for the feedback we received from
various participants at 
EC'13 and the COST COMSOC Summer School on Matching '13. 
Any errors remain our own.
Part of this research was supported by research grants from the Hasler Foundation and from the Swiss National Science Foundation.}}}
\author{%
Timo Mennle \\ University of Zurich
\and Sven Seuken \\ University of Zurich }
\date{First version: February 13, 2013 \\
This version: \today}
\maketitle
\author{\setcounter{footnote}{2}Timo Mennle \and Sven Seuken}
\date{\today}

\vspace{-1em} 
\begin{abstract}
Severe impossibility results restrict the design of strategyproof random assignment mechanisms, and 
trade-offs are necessary when aiming for more demanding efficiency requirements, such as ordinal or rank efficiency. 
We introduce \emph{hybrid mechanisms}, which are convex combinations of two component mechanisms. 
We give a set of conditions under which such hybrids facilitate a non-degenerate trade-off between strategyproofness (in terms of \emph{partial strategyproofness}) and efficiency (in terms of \emph{dominance}). 
This set of conditions is tight in the sense that trade-offs may become degenerate if any of the conditions are dropped. 
Moreover, we give an algorithm for the mechanism designer's problem of determining a maximal mixing factor. 
Finally, we prove that our construction can be applied to mix Random Serial Dictatorship with Probabilistic Serial, as well as with the adaptive Boston mechanism, and we illustrate the efficiency gains numerically. 
\end{abstract}

}

\section{Introduction}
\label{sec:intro}
\enlargethispage{1em}
When a set of indivisible goods or resources (called \emph{objects}) has to be assigned to self-interested
agents without the use of monetary transfers, we face an \emph{assignment problem}.
Examples include the assignment of students to schools, subsidized housing to tenants, and teachers to
training programs \citep{Niederle2008}. 
Since the seminal paper of \citet{Hylland1979EfficientAllocation}, this problem has attracted much attention from mechanism designers (e.g., \citet{AbdulSonmez1998RSDandTheCore,Bogomolnaia2001ANewSolution,AbdulSonmez2003OrdinalEffAndDomSets,Featherstone2011RankBasedRefinementWP}). 

It is often desirable or even required that assignment mechanisms perform well on
multiple dimensions, such as \emph{efficiency}, \emph{fairness}, and \emph{strategyproofness}. 
However, severe impossibility results prevent the design of mechanisms that excel on all these dimensions simultaneously \citep{Zhou1990ImpossibilityOneSidedMatching}. 
This makes the assignment problem a challenge for mechanism designs. 
The \emph{Random Serial Dictatorship (RSD)} mechanism is \emph{strategyproof} and \emph{anonymous},
but only satisfies the baseline requirement of \emph{ex-post efficiency}. 
If strategyproofness is relaxed to \emph{weak strategyproofness}, the more demanding requirement of \emph{ordinal efficiency} can be achieved via the \emph{Probabilistic Serial (PS)} mechanism. 
However, no ordinally efficient, symmetric mechanism can be strategyproof \citep{Bogomolnaia2001ANewSolution}. 
The even more demanding requirement of \emph{rank efficiency} can be achieved by \emph{Rank Value (RV)} mechanisms, but at the same time, rank efficiency is in conflict with even weak strategyproofness. 
Obviously, trade-offs are necessary and have been the focus of recent research (e.g., see \citep{AbdulPathakRoth2009SPvsEffWithIndifferencesNYCSchoolMatch,
Budish2012MatchingVersusMechanismDesign,
AzizBrandBrill2013OnTradeoffEffAndSPInRandSocialChoice,
AzevedoBudish2015SPL}). 

In this paper we investigate a straightforward approach to the problem of trading off strategyproofness and efficiency of random assignment mechanisms. 
Specifically, we use \emph{partial strategyproofness} \citep{MennleSeuken2017PSP_WP} to compare mechanisms by their incentive properties, and we use \emph{ordinal dominance} and \emph{rank dominance} to compare them by their efficiency properties. 
We introduce \emph{hybrid mechanisms}, which are convex combinations of two component mechanisms, and we show that, subject to a set of quite intuitive conditions, hybrid mechanisms behave exactly as one
would expect: 
they facilitate a non-degenerate trade-off between strategyproofness and efficiency. 
Instantiating this approach with popular assignment mechanisms, such as RSD, PS, RV, and variants of the Boston mechanism, we illustrate that the conditions are not trivial; but when they hold, the efficiency gains (over RSD) can be substantial. 
\subsection{Partially Strategyproof Hybrid Mechanisms}
\label{sec:intro:psp_hybrids}
Due to restrictive impossibility results pertaining to strategyproofness, we cannot hope to improve efficiency of mechanisms without relaxing strategyproofness, especially in the presence of additional fairness criteria, such as anonymity. 
In \citep{MennleSeuken2017PSP_WP} we have introduced a relaxed incentive requirement for assignment mechanisms: 
a mechanism is \emph{$r$-partially strategyproof} if truthful reporting is a dominant strategy for agents who have sufficiently different valuations for different objects. 
The numerical parameter $r$ controls the extent to which their valuations must vary across objects. 
$r$ yields a parametric measure for the strength of the incentive properties of non-strategyproof mechanisms. 
Larger values of $r \in [0,1]$ imply stronger incentive guarantees, $r=1$ corresponds to strategyproofness, and $r = 0$ does not yield any incentive guarantees. 
We use this \emph{degree of strategyproofness} to quantify the performance of mechanisms on the strategyproofness-dimensions. 
In this paper, we study how hybrid mechanisms trade off strategyproofness and efficiency. 
The key idea of \emph{hybrid mechanisms} is to ``mix'' a mechanism $\f$ that has good incentive properties and another mechanism $\g$ that has good efficiency properties. 
Concretely, for two component mechanisms $\f$ and $\g$ the $\beta$-hybrid is given by $\h = \hyb$. The parameter $\beta \in [0,1]$ is called the \emph{mixing factor}. 
Obviously, $h^0 = \f$ and $h^1 = \g$, so that the hybrid mechanisms at the extreme mixing factors $\beta = 0$ and $\beta = 1$ trivially inherit the desirable property of the respective component mechanism. 
For intermediate mixing factors $\beta \in (0,1)$ hybrids should intuitively inherit a \emph{share} of the desirable properties from both component mechanisms. 

Regarding the strategyproofness-dimension, we find that this intuition may not always be justified: 
as we show in this paper, it can happen that any non-trivial share of $\g$ (i.e., any $\beta > 0$) causes the degree of strategyproofness of the hybrid to drop to $0$ immediately. 
Our first main result is a set of sufficient conditions that prevent such ``degenerate'' behavior: 
a pair $(\f,\g)$ is \emph{hybrid-admissible} if 
\vspace{-0.4em}
\begin{enumerate}[(1)]
\setlength{\itemsep}{0pt}
	\item $\f$ is strategyproof, 
	\item $\g$ is \emph{upper invariant}: a swap of two adjacent objects in an agent's report does not change that agent's probabilities for obtaining an object it prefers to any of the two, 
	\item $\g$ is \emph{weakly less varying} than $\f$: whenever a swap leads to a change of an agent's assignment under $\g$, then that agent's assignment must also change under $\f$.
\end{enumerate}
For any hybrid-admissible pair, we show that a non-degenerate trade-off is possible in the sense that for all (even small) relaxations of strategyproofness, the share of $\g$ that can be included in the hybrid is non-trivial. 
Furthermore, we show that it is not possible to drop any of the the three conditions from hybrid-admissibility and still guarantee that hybrid mechanisms with intermediate degrees of strategyproofness can be constructed. 
\subsection{Harnessing Efficiency Improvements}
\label{sec:intro:efficiency}
Towards understanding the efficiency improvements that we can obtain through hybrid mechanisms, we
employ the well-known concepts of \emph{ordinal} and \emph{rank dominance}.
Our second main result is that if $\g$ dominates $\f$, then $\h$ dominates $\f$ but is dominated by $\g$. 
Thus, hybrids have \emph{intermediate efficiency} in a well-defined sense.  
One challenge is that a comparison of $\f$ and $\g$ by dominance may not be possible at every preference profile. 
For cases where they are incomparable at some preference profile, we show that $\h$ is not dominated by $\f$ and $\h$ does not dominate $\g$. 
In other words, the dominance comparison of $\h$, $\f$, and $\g$ \emph{points in the right direction} whenever this comparison is possible, and it \emph{never points in the wrong direction} when $\f$ and $\g$ are not comparable. 

This shows that if some mechanism $\f$ offers good incentives, and another mechanism $\g$ has desirable efficiency, then a mechanism designer can trade off strategyproofness and efficiency systematically by constructing hybrids of $\f$ and $\g$. 
Concretely, she can specify the minimal acceptable degree of strategyproofness $\underline\rho$ and then choose the mixing factor $\beta$ as high as possible. 
The \emph{maximal mixing factor} $\bmax$ is the largest value of $\beta \in [0,1]$ for which $\h$ is $\underline\rho$-partially strategyproof.
The parameter $\bmax$ has an appealing interpretation: 
it serves as a measure for how far relaxing strategyproofness from ``$r = 1$'' (i.e., strategyproofness) to ``$r \geq \underline\rho$'' (i.e., $\underline\rho$-partial strategyproofness) will take us between the baseline efficiency of $\f$ to the more desirable efficiency of $\g$. 

This ``trade-off'' approach to the design of random assignment mechanisms gives rise to a computational problem: 
given a setting (i.e., number of agents, number of objects, and object capacities), as well as a minimal acceptable degree of strategyproofness $\underline\rho$, the mechanism designer faces the problem of determining the maximal mixing factor $\bmax$. 
In this paper, we show how this problem can be solved algorithmically for any finite setting. 

\subsection{Hybrids of Popular Mechanisms}
\label{sec:intro:application}
Finally, we apply our theory of hybrids to pairs of popular mechanisms. 
First, we show that Random Serial Dictatorship (RSD) and the Probabilistic Serial (PS) mechanism form a hybrid-admissible pair. 
Since PS is ordinally efficient, it ordinally dominates RSD whenever the two mechanisms are comparable. 
Therefore, hybrids of RSD and PS can be used to trade off strategyproofness and efficiency in terms of ordinal dominance. 

Second, we show two impossibility result: neither the classic Boston mechanism (NBM),\footnote{We call this variant the \emph{na\"ive} Boston mechanism because it is ``na\"ive'' in the sense that it lets agents apply to exhausted objects in the application process \citep{MennleSeuken2017ABM_WP}.} nor Rank Value (RV) mechanisms are weakly less varying than RSD. 
Furthermore, we demonstrate that hybrids of RSD with NBM or RV indeed have degenerate incentive properties (i.e., they have a degree of strategyproofness of $0$). 
These findings illustrate that while hybrid-admissibility is sufficient for non-degenerate trade-offs, it is also close to being necessary. 
On a broader scale, these impossibility results serve as reminders that straightforward approaches, like the construction of hybrid mechanisms, do not always yield the seemingly obvious outcomes that our intuition may suggest. 

Third, we show that the pair of RSD and the adaptive Boston Mechanism (ABM)\footnote{ABM is a variant of the Boston school choice mechanisms in which students automatically skip exhausted schools in the application process; see \citep{MennleSeuken2017ABM_WP}.} 
is also hybrid-admissible. 
ABM rank dominates RSD whenever the outcomes are comparable, except in a negligibly small number of cases \citep{MennleSeuken2017ABM_WP}. 
Thus, hybrids of RSD and ABM allow non-degenerate trade-offs between strategyproofness and efficiency in terms of the rank dominance relation (except for the small number of cases). 
For both pairs (RSD,PS) and (RSD,ABM), we find numerically that efficiency gains (in terms of the maximal mixing factor) from relaxing strategyproofness can be substantial. 

%
%
%
%
%
%
%
%
%
%
%
%
%
%
%
%
%
%
%
%
%

\medskip
\noindent
\textbf{Organization of this paper:} 
In Sections \ref{sec:related} and \ref{sec:model}, we discuss related work and introduce our formal model. 
In Section \ref{sec:psphybrids}, we introduce hybrid-admissibility and show that it enables the construction of partially strategyproof hybrids.\footnote{We emphasize that the partial strategyproofness concept imported from \citep{MennleSeuken2017PSP_WP} in Section \ref{sec:psphybrids:spconcepts} should \emph{not} be considered a contribution of the present paper.}
In Section \ref{sec:efficiency}, we show how hybrids trade off strategyproofness and efficiency, and we give an algorithm for the mechanism designer's problem of determining a maximal mixing factor. 
In Section \ref{sec:applications}, we apply our results to popular assignment mechanisms. 
Section \ref{sec:conclusion} concludes.

\section{Related Work}
\label{sec:related}
\citet{Hylland1979EfficientAllocation} proposed a pseudo-market mechanism for the problem of assigning students to on-campus housing. 
However, eliciting agents' cardinal preferences is often difficult if not impossible to in settings without money. 
For this reason, recent work has focused on \emph{ordinal} mechanisms, where agents submit rankings over objects. 
\citet{Carroll2011ElicitingOrdinalPreferences}, \citet{HuesmanWambach2015MatchingWithMoralConcerns}, and \citet{EhlersMajumdarMishraSen2016ContinuityandIncentiveCompatibilityinCardinalVotingMechanisms} gave systematic arguments for the focus on ordinal mechanisms. 

For the deterministic case, strategyproofness of assignment mechanisms has been studied extensively, e.g., in \citep{Papai2000SPAssignmentHierarchicalExchange,EhlersKlaus2006ReallocConsistent,KlausEhlers2007Consistent,Hatfield2009StrategyproofEfficientAndNonbossyQuotaAllocations,PyciaUenver2014ICallocation}. 
For non-deterministic mechanisms, \citet{AbdulSonmez1998RSDandTheCore} showed that RSD is equivalent to the Core from Random Endowments mechanism for house allocation (if agents' initial houses are drawn uniformly at random). 
\citet{Erdil2014SPStochasticAssignment} showed that when capacity exceeds demand, RSD is not the only strategyproof, ex-post efficient mechanism that satisfies symmetry. 
On the other hand, \citet{Bade2014RSDOneAndOnly} showed that taking any ex-post efficient, strategyproof, non-bossy, deterministic mechanism and assigning agents to roles in the mechanism uniformly at random is equivalent to using RSD. 
However, when capacity equals demand, it is still an open conjecture whether RSD is the unique mechanism that is strategyproof, ex-post efficient, and symmetric \citep{LeeSethuraman2011EquivalenceResultsAllocationUnifiedView}. 
Despite the fact that this conjecture remains to be proven, is evident that the space of ``useful'' strategyproof mechanisms is very small. 

The research community has also introduced stronger efficiency concepts, such as \emph{ordinal efficiency}. 
The Probabilistic Serial (PS) mechanism, which achieves ordinal efficiency, was originally introduced by \citet{Bogomolnaia2001ANewSolution} for strict preferences and since then it has been studied extensively: 
\citet{KattaSethuraman2006} introduced an extension that allows agents to be indifferent between objects. 
\citet{Hashimoto2013TwoAxiomaticApproachestoPSMechanisms} showed that PS with equal eating speads is the unique mechanism that is ordinally fair and non-wasteful. 
In terms of incentives, \citet{Bogomolnaia2001ANewSolution} showed that PS is not strategyproof but \emph{weakly strategyproof} in the sense that no agent can obtain a first order-stochastically dominant assignment by misreporting. 

While ex-post efficiency and ordinal efficiency are the two most well-studied efficiency concepts for assignment mechanisms, some mechanisms used in practice aim at \emph{rank efficiency}, which is a further refinement of ordinal efficiency \citep{Featherstone2011RankBasedRefinementWP}. 
However, no rank efficient mechanism can even be weakly strategyproof. 
Other popular mechanisms, like the Boston mechanism (see \citep{ErginSonmez2006GamesUnderBostonSchoolChoice}), are manipulable but are nevertheless in frequent use. 
\citet{Budish2012TheMultiUniAssignmentProblemTheoryAndEvidence} showed evidence from course allocation at Harvard Business School, suggesting that using a non-strategyproof mechanism may lead to higher social welfare than using a strategyproof mechanism such as RSD. 
This challenges the view that strategyproofness should be a hard requirement for mechanism design. 

Given that strategyproofness is so restrictive, some researchers have considered relaxed incentive requirements. 
For example, \citet{Carroll2013AQuantitativeApproachToIncentives} used approximate strategyproofness for normalized vNM utilities in the voting domain to quantify the incentives to manipulate under different non-strategyproof voting rules. 
\citet{Budish2011TheCombinatorialAssignmentProblem} proposed the Competitive Equilibrium from Approximately Equal Incomes mechanism for the combinatorial assignment problem. 
For the random social choice domain, \citet{AzizBrandBrill2013OnTradeoffEffAndSPInRandSocialChoice} considered first order-stochastic dominance (SD) and sure thing (ST) dominance. 
They showed that while RSD is SD-strategyproof, it is merely ST-efficient; 
they contrasted this with Strictly Maximal Lotteries, which are SD-efficient but only ST-strategyproof. 

The construction of hybrid mechanisms in the present paper differs from these approaches: 
rather than comparing discrete points in the design space, we enable a continuous trade-off between
strategyproofness and efficiency that can be described in terms of two parameters: 
the \emph{degree of strategyproofness} \citep{MennleSeuken2017PSP_WP} for incentive properties and the \emph{mixing factor} for efficiency. 
Formally, hybrid mechanisms are simply probability mixtures of two component mechanisms. 
\citet{Gibbard1977ManipulationOfVotingWithChance} used such mixtures in his seminal characterization of the set of strategyproof random ordinal mechanisms. 
In \citep{MennleSeuken2017EFF_WP}, we have extended \citeauthor{Gibbard1977ManipulationOfVotingWithChance}'s result by giving a structural characterization of the Pareto frontier of \emph{approximately} strategyproof random mechanisms. 
Hybrid mechanisms play a central role in our characterization. 
The present paper differs from \citep{MennleSeuken2017EFF_WP} in that we consider the random assignment domain specifically, and we employ \emph{partial strategyproofness}, which is a more appropriate relaxation of strategyproofness in this domain than approximate strategyproofness. 
\section{Formal Model}
\label{sec:model}
A \emph{setting} \NMQ\ consists of a set $N$ of $n$ \emph{agents}, a set $M$ of $m$ \emph{objects}, and a vector $\bm q = (q_1, \ldots, q_m)$ of \emph{capacities} (i.e., there are $q_j$ units of object $j$).
We assume that supply satisfies demand (i.e., $n \leq \sum_{j \in M} q_j$), since we can always add a dummy object with capacity $n$.
Each agent $i \in N$ has a strict preference order $P_i$ over objects, where $P_i:a\succ b$ means that $i$ strictly prefers object $a$ to object $b$. 
We denote the set of all preference orders by $\mathcal{P}$.
A \emph{preference profile} $\bm P = (P_1 , \ldots, P_n) \in \mathcal{P}^N$ is a collection of preference orders of all agents, where $P_{-i} \in \mathcal{P}^{N\backslash\{i\}}$ are the preference orders of all agents, except $i$. 
Agents' preferences over objects are extended to preferences over lotteries via \emph{von Neumann-Morgenstern (vNM) utilities} $u_i:M\rightarrow\mathbb{R}^+$.
A utility function $u_i$ is \emph{consistent} with preference order $P_i$ if $P_i:a \succ b$ whenever $u_i(a) > u_i(b)$. 
We denote by $U_{P_i}$ the set of all utility functions that are consistent with $P_i$.

In the \emph{random assignment} problem, each agent ultimately receives a single object, but we evaluate random mechanisms based on the resulting \emph{interim assignments}. 
Such assignments are represented by an $n\times m$-matrix $x = (x_{i,j})_{i \in N, j \in M}$ satisfying the \emph{fulfillment constraint} $\sum_{j\in M} x_{i,j} = 1$, the \emph{capacity constraint} $\sum_{i\in N} x_{i,j} \leq q_j$, and the \emph{probability constraint} $x_{i,j} \in [0,1]$ for all $i\in N,j \in M$. 
The entries of the matrix $x$ are interpreted as probabilities, where $x_{i,j}$ is the probability that $i$ gets $j$.
An assignment is \emph{deterministic} if all agents get exactly one full object, such that $x_{i,j} \in \{0,1\}$ for all $i \in N, j \in M$.
For any agent $i$, the $i$th row $x_i = (x_{i,j})_{j \in M}$ of the matrix $x$ is called the \emph{assignment vector of $i$}, or \emph{$i$'s assignment} for short.
The Birkhoff-von Neumann Theorem \citep{Birkhoff1946,Neumann1953} and its extensions \citep{Budishetal2013DesignRandomAllocMechsTheoryAndApp} ensure that, given any probabilistic assignment, we can always find a lottery over deterministic assignments that implements its marginal probabilities.
Finally, we denote by $X$ and $\Delta(X)$ the spaces of all deterministic and probabilistic assignments, respectively.

A \emph{mechanism} is a mapping $ \varphi: \mathcal{P}^N \rightarrow \Delta(X)$ that chooses an assignment based on a profile of reported preference orders.
$\varphi_i(P_i,P_{-i})$ is the assignment vector that agent $i$ receives if it reports $P_i$ and the other agents report $P_{-i}$.
Note that mechanisms only receive rank ordered lists as input but no additional cardinal information.
Thus, we consider \emph{ordinal} mechanisms, which determine assignments based on the ordinal preference reports alone.
The expected utility for $i$ is given by the scalar product
\begin{equation}
	\mathbb{E}_{\varphi_i(P_i,P_{-i})}[u_i] = \left\langle u_i,\varphi_i(P_i,P_{-i})\right\rangle = \sum_{j \in M} u_i(j) \cdot \varphi_{i,j}(P_i,P_{-i}).
\end{equation}
Finally, we define hybrid mechanisms, which we study in this paper. 
\enlargethispage{1em}
\begin{definition}[Hybrid]
For mechanisms $\f,\g$ and a \emph{mixing factor} $\beta \in [0,1]$, the \emph{$\beta$-hybrid of $\f$ and $\g$} is given by $\h = \hyb$, where for all preference profiles $\bm P \in \mathcal{P}^N$, the assignment $\h(\bm P)$ is the $\beta$-convex combination of the assignments of $\f(\bm P)$ and $\g(\bm P)$. 
%
%
\end{definition}
\section{Partially Strategyproof Hybrid Mechanisms}
\label{sec:psphybrids}
In this section, we first provide a short overview of the partial strategyproofness concept, which we have introduced in \citep{MennleSeuken2017PSP_WP}. 
We then give our first main result, a set of conditions under which the construction of hybrid mechanisms with non-degenerate degrees of strategyproofness is possible. 
Subsequently, we show that none of the conditions can be dropped without losing this guarantee. 
\subsection{Full and Partial Strategyproofness} 
\label{sec:psphybrids:spconcepts}
Under a strategyproof mechanism, agents have a dominant strategy to report truthfully. 
For random mechanisms, this means that truthful reporting of ordinal preferences maximizes any agent's expected utility, independent of the reports of the other agents and the particular utility function underlying that agent's preference order. 
\begin{definition}[Strategyproof] 
\label{def:sp}
A mechanism $\f$ is \emph{strategyproof} if for any agent $i\in N$, any preference profile $(P_i,P_{-i})\in \mathcal{P}^N$, any misreport $P_i' \in \mathcal{P}$, and any utility function $u_i \in U_{P_i}$ that is consistent with $P_i$, 
we have 
\begin{equation}
	\left\langle u_i,\f_i(P_i,P_{-i}) - \f_j(P_i',P_{-i}) \right\rangle \geq 0
\end{equation}
\end{definition}
This notion of strategyproofness for random mechanism coincides with the one used by \citet{Gibbard1977ManipulationOfVotingWithChance} for random voting mechanisms. 
For deterministic mechanisms, it reduces to the requirement that no agent can obtain a strictly preferred object by misreporting. 
Furthermore, it is equivalent to \emph{strong stochastic dominance-strategyproofness}, which requires that any agent's assignment from misreporting is first order-stochastically dominated by the assignment that the agent can obtain from reporting truthfully. 

\emph{Partially strategyproof} mechanisms \citep{MennleSeuken2017PSP_WP} remain strategyproof on a particular domain restriction. 
The agents can still have any preference order, but their underlying utility functions are constrained. 
\begin{definition}[Uniformly Relatively Bounded Indifference] 
\label{def:urbi}
For $r \in [0,1]$, a utility function $u \in U_P$ satisfies \emph{uniformly relatively bounded indifference} with respect to \emph{indifference bound} $r$ (\URBIr) if for any objects $a,b \in M$ with $P: a \succ b$, we have
\begin{equation}
		r\cdot\left(u(a) - \min_{j \in M} u(j)\right) \geq u(b) - \min_{j \in M} u(j).
	\label{eq:urbi}
\end{equation}
\end{definition}
\enlargethispage{1em}
To obtain some intuition about this domain restriction, observe that a utility function $u: M \rightarrow \mathbb{R}^+$ can be interpreted as a vector in $\left(\mathbb{R}^+\right)^m$.
The set $U_P$ corresponds to a convex cone containing all the vectors for which the $a$-component is strictly larger than the $b$-component (provided $P:a\succ b$). 
Then the set of utility functions that satisfy \URBIr\ and are consistent with $P$ corresponds to a smaller cone inside $U_P$. 
This smaller cone is strictly bounded away from the indifference hyperplanes $H_{a,b} = \{u(a) = u(b)\}$ for any two objects $a,b \in M$. 
Note that the \URBIr\ constraint is independent of affine transformations: 
if $u$ is translated by adding a constant $\delta$ (i.e., $\tilde{u}(j) = u(j) + \delta$ for all $j\in M$), then this value will be subtracted again in (\ref{eq:urbi}), since $\min_{j \in M} \tilde{u}(j) = \min_{j\in M} u(j) + \delta$. 
If $u$ is scaled by a factor $\alpha > 0$, then this affects both sides of (\ref{eq:urbi}) equally, so that the \emph{relative} bound $r$ is preserved. 

For convenience, we denote by \URBIr\ the set of all utility functions that satisfy uniformly relatively bounded indifference with respect to $r$. 
With this domain restriction, the definition of partial strategyproofness is analogous to the definition of strategyproofness, except that the inequality only needs to hold for agents with utility functions in \URBIr. 
\begin{definition}[Partially Strategyproof]
\label{def:psp}
For a given setting \NMQ\ and $r \in [0,1]$ we say that a mechanism $\f$ is \emph{$r$-partially strategyproof in \NMQ} if for any agent $i\in N$, any preference profile $(P_i,P_{-i})\in \mathcal{P}^N$, any misreport $P_i' \in \mathcal{P}$, and any utility function $u_i \in U_{P_i} \cap \URBIr$ that is consistent with $P_i$ and satisfies \URBIr, 
we have 
\begin{equation}
	\left\langle u_i,\f_i(P_i,P_{-i}) - \f_j(P_i',P_{-i}) \right\rangle \geq 0.
\end{equation}
\end{definition}
For the remainder of the paper we will fix an arbitrary setting \NMQ. 
Thus, we will simply say that a mechanism is \emph{$r$-partially strategyproof}, omitting the setting but keeping in mind that the value of $r$ is specific to the respective setting. 

One of the main findings in \citep{MennleSeuken2017PSP_WP} is that strategyproofness can be decomposed into three simple axioms, and that the set of partially strategyproof mechanisms arises by dropping the least important of these axioms.
The three axioms restrict the way in which a mechanism may change an agent's assignment when this agent changes its report. 
For any preference order $P \in \mathcal{P}$, its \emph{neighborhood} $N_P$ is the set of preference orders that can be obtained by swapping two objects that are adjacent in $P$. 
Formally, for $P: a_1 \succ \ldots\succ a_m$ we have 
\begin{equation}
	N_P = \left\{P' \in \mathcal{P} ~\left|~ 
		\begin{array}{l}
			P':a_1 \succ \ldots\succ a_{k-1} \succ a_{k+1} \succ a_k \succ a_{k+2} \succ \ldots \succ a_m \\	
			\text{ for some }k \in \{1,\ldots,m-1\}
		\end{array}\right.\right\}.
\end{equation}
\begin{definition}[Swap Monotonic]
\label{def:sm}
A mechanism $\f$ is \emph{swap monotonic} if for any agent $i\in N$, any preference profile $(P_i,P_{-i}) \in \mathcal{P}^N$, and any misreport $P_i' \in N_{P_i}$ from the neighborhood of $P_i$ with $P_i: a \succ b$ and $P_i': b \succ a$, one of the following holds: 
\vspace{-0.4em}
\begin{itemize}
\setlength{\itemsep}{0pt}
	\item either $\f_i(P_i,P_{-i}) = \f_i(P_i',P_{-i})$, 
	\item or $\f_{i,a}(P_i,P_{-i}) > \f_{i,a}(P_i',P_{-i})$ and $\f_{i,b}(P_i,P_{-i}) < \f_{i,b}(P_i',P_{-i})$.
\end{itemize}
\end{definition}
\begin{definition}[Upper Invariant]
\label{def:ui}
A mechanism $\f$ is \emph{upper invariant} if for any agent $i\in N$, any preference profile $(P_i,P_{-i}) \in \mathcal{P}^N$, and any misreport $P_i' \in N_{P_i}$ from the neighborhood of $P_i$ with $P_i: a \succ b$ and $P_i': b \succ a$, we have that $i$'s assignment for objects from the upper contour set of $a$ does not change (i.e., $\f_{i,j}(P_i,P_{-i}) = \f_{i,j}(P_i',P_{-i})$ for all $j\in M$ with $P_i : j \succ a$).
\end{definition}
\begin{definition}[Lower Invariant]
\label{def:li}
A mechanism $\f$ is \emph{lower invariant} if for any agent $i\in N$, any preference profile $(P_i,P_{-i}) \in \mathcal{P}^N$, and any misreport $P_i' \in N_{P_i}$ from the neighborhood of $P_i$ with $P_i: a \succ b$ and $P_i': b \succ a$, we have that $i$'s assignment for objects from the lower contour set of $b$ does not change (i.e., $\f_{i,j}(P_i,P_{-i}) = \f_{i,j}(P_i',P_{-i})$ for all $j\in M$ with $P_i : b \succ j$).
\end{definition}
Swap monotonicity means that if the mechanism changes an agent's assignment after this agent has swapped two adjacent objects in its report, then this change of assignment must be \emph{direct} and \emph{responsive}: 
if there is any change at all, there must be some change for the objects for which differential preferences have been reported, and this change has to be in the right direction. 
Upper invariance means that an agent cannot improve its chances at more preferred objects by changing the order of less preferred objects. 
In the presence of an outside option, this is equivalent to robustness to manipulation by truncation \citep{Hashimoto2013TwoAxiomaticApproachestoPSMechanisms}. 
Finally, lower invariance is the natural counterpart for upper invariance. 
Strategyproofness decomposes into these three axioms, and partial strategyproofness arises by dropping lower invariance. 
\begin{fact}[\citeauthor{MennleSeuken2017PSP_WP}, \citeyear{MennleSeuken2017PSP_WP}] 
\label{fact:thm1}
A mechanism is strategyproof \emph{if and only if} it is swap monotonic, upper invariant, and lower invariant. 
\end{fact}
\begin{fact}[\citeauthor{MennleSeuken2017PSP_WP}, \citeyear{MennleSeuken2017PSP_WP}] 
\label{fact:thm2}
Given a setting \NMQ, a mechanism is $r$-partially strategyproof for some $r>0$ \emph{if and only if} it is swap monotonic and upper invariant. 
\end{fact}
Furthermore, the \URBIr\ domain restriction is \emph{maximal} in the sense that for a swap monotonic, upper invariant mechanism, there is no systematically larger set of utility functions for which we can also guarantee that truthful reporting is a dominant strategy. 
This allows us to define a meaningful, parametric measure for the incentive properties of non-strategyproof mechanisms.
\begin{definition}[Degree of Strategyproofness]
\label{def:dosp}
Given a setting \NMQ\ and a mechanism $\f$, the \emph{degree of strategyproof} of $\f$ is the largest indifference bound $r \in [0,1]$ for which $\f$ is $r$-partially strategyproof. 
Formally, 
\begin{equation}
	\rho_{(N,M,\bm q)}(\f) = \max \{r \in [0,1] ~|~\f\text{ is $r$-partially strategyproof}\}.
\end{equation}
\end{definition}
By virtue of the maximality of the \URBIr\ domain restriction, it is meaningful to compare mechanisms by their degree of strategyproofness. 
This comparison is consistent with (but not equivalent to) the comparison of mechanisms by their vulnerability to manipulation \citep{Pathak2013AERComparingMechanismsByVulnerability}. 
In this paper, we use the degree of strategyproofness to measure and compare the performance of mechanisms on the strategyproofness-dimension. 
\subsection{Construction of Partially Strategyproof Hybrids}
\label{sec:psphybrids:construction}
To state our first main result, we define what it means for one mechanism $\g$ to be \emph{weakly less varying} than another mechanism $\f$. 
This condition is part of our subsequent definition of \emph{hybrid-admissibility}. 
\begin{definition}[Weakly Less Varying]
\label{def:wlv}
For mechanisms $\f,\g$, we say that \emph{$\g$ is weakly less varying than $\f$} if for any agent $i\in N$, any preference profile $(P_i,P_{-i}) \in \mathcal{P}^N$, and any report $P_i' \in N_{P_i}$ from the neighborhood of $P_i$ we have that 
\begin{equation}
	\f_i(P_i,P_{-i}) = \f_i(P_i',P_{-i})~\Rightarrow~\g_i(P_i,P_{-i}) = \g_i(P_i',P_{-i}).
\end{equation}
\end{definition}
Loosely speaking, this means that the mechanism $\g$ (as a function of preference profiles) must be at least as \emph{coarse} as $\f$. 
If $\f$ does not change $i$'s assignment, then a weakly less varying mechanism $\g$ must not change it either.  
This is important for the incentive properties of hybrids: 
suppose that some misreport by some agent is beneficial under $\g$. 
If for the same misreport, $\f$ does not change that agent's assignment, then any share of $\f$ in the hybrid is insufficient to counteract the benefit that the agent obtains from this manipulation. 

We are now ready to formulate hybrid-admissibility. 
\begin{definition}[Hybrid-Admissible]
\label{def:hybridadmissible}
A pair $(\f,\g)$ is \emph{hybrid-admissible} if 
\vspace{-0.4em}
\begin{enumerate}[(1)]
\setlength{\itemsep}{0pt}
	\item $\f$ is strategyproof,
	\item $\g$ is upper invariant, 
	\item $\g$ is weakly less varying than $\f$.
\end{enumerate}
\end{definition}
The following Theorem \ref{thm:construction} is our first main result. 
It shows that under hybrid-admissibility, the degree of strategyproofness $\rho(\h)$ of hybrid mechanisms varies in a non-degenerate fashion for varying mixing factors $\beta \in [0,1]$. 
\begin{theorem}
\label{thm:construction}
\ThmConstructionStatement
\end{theorem}
\begin{proof}[Proof Outline (formal proof in Appendix \ref{app:proofs:construction})]
Consider agent $i$ with $P_i: a_1 \succ \ldots \succ a_m$ and a misreport $	P_i: a_1 \succ \ldots \succ a_K \succ a_{K+1}' \succ \ldots \succ a_m'$, where the positions of the first $K$ objects remain unchanged. 
The key insight is that we only need to consider cases where $i$'s assignment of $a_{K+1}$ \emph{strictly decreases} under $\f$. 
If $i$ receives less of $a_{K+1}$, this has a negative effect on $i$'s expected utility from reporting $P_i'$. 
We show that for utility functions in \URBIr\ and sufficiently small $\beta > 0$, this negative effect suffices to make the misreport $P_i'$ useless. 
Finally, $\beta > 0$ can be chosen uniformly for all preference profiles and misreports (while it may depend on the mechanisms and the setting).
\end{proof}

Theorem \ref{thm:construction} confirms our intuition about the manipulability of hybrids $\h$ when $\beta$ varies between 0 and 1. 
For mechanism designers, this result is good news: 
given any setting, a hybrid-admissible pair of mechanisms and a minimal acceptable degree of strategyproofness $\underline\rho \in [0,1)$, we can always find a non-trivial hybrid (i.e., $\h$ with $\beta > 0$) that is $\underline\rho$-partially strategyproof. 
The fact that a \emph{strictly positive} $\beta$ can be chosen implies that any (even small) relaxation of strategyproofness can enable improvements on the efficiency-dimension. 

If $\g$ is a more efficient mechanism, then a mechanism designer would intuitively like to choose a mixing factor as large as possible because more of the more efficient $\g$ would be included. 
In Section \ref{sec:efficiency} we give a precise understanding of the way in which mixing affects the efficiency of hybrids, and we show that the mechanism designer's problem of determining a \emph{maximal mixing factor} can be solved algorithmically. 
\subsection{Independence of Hybrid-Admissibility for Theorem \ref{thm:construction}}
\label{sec:psphybrids:tightness}
We have seen that under hybrid-admissibility, the degree of strategyproofness of hybrid mechanisms in fact behave as our intuition suggests.  
Next, we show that dropping either of the three conditions from hybrid-admissibility will lead to a collapse of this guarantee. 
\begin{proposition}
\label{prop:tightness:fsp}
If $\f$ is not strategyproof, there exists a mechanism $\g$ that is upper invariant and weakly less varying than $\f$, and a bound $r \in (0,1)$, such that no non-trivial hybrid of the pair $(\f,\g)$ will be $r$-partially strategyproof.
\end{proposition}
\begin{proof}
Consider a constant mechanism $\g$ that yields the same assignment, independent of the agents' reports. If $\f$ is manipulable by some agent $i$ with utility $u_i$, we choose $r$ such that $u_i \in \URBIr$. 
Then $i$ can manipulate any non-trivial hybrid of $\f$ and $\g$.
\end{proof}
\begin{proposition}
\label{prop:tightness:gui}
For any strategyproof $\f$ and any $\g$ that is weakly less varying than $\f$, but not upper invariant, no non-trivial hybrid of the pair $(\f,\g)$ is $r$-partially strategyproof for any bound $r \in (0,1]$.
\end{proposition}
\begin{proof} 
Since $\f$ is strategyproof, it must be upper invariant by Fact \ref{fact:thm1}. 
If $\g$ is not upper invariant, then neither is any non-trivial hybrid of $\f$ and $\g$. 
Consequently, the hybrid is not $r$-partially strategyproof for any $r > 0$ by Fact \ref{fact:thm2}. 
\end{proof}
\begin{proposition}
\label{prop:tightness:gwlvf}
Let $\f$ be a strategyproof mechanism such that for some agent $i\in N$, some preference profile $(P_i,P_{-i})$, and some misreport $P_{i}'$, we have $\f_i(P_i,P_{-i}) = \f_i(P_i',P_{-i})$. 
Then there exits an upper invariant mechanism $\g$ such that no non-trivial hybrid of the pair $(\f,\g)$ is $r$-partially strategyproof for any bound $r \in (0,1]$
\end{proposition}
\begin{proof} 
Let $j$ be the best choice object under $P_i$ that changes position between $P_i$ and $P_i'$. 
Then let $\g$ be upper invariant with $\g_{i,j}(P_i,P_{-i}) = 0$ and $\g_{i,j}(P_i',P_{-i}) = 1$. 
If $\beta > 0$, then $i$ can manipulate the hybrid $\h$ in a first order-stochastic dominance sense. 
Therefore, $\h$ cannot be partially strategyproof by Proposition 2 in \citep{MennleSeuken2017PSP_WP}.
\end{proof}
In combination, Propositions \ref{prop:tightness:fsp}, \ref{prop:tightness:gui}, and \ref{prop:tightness:gwlvf} show that none of the three requirements for hybrid admissibility can be dropped, or else the relaxed incentive properties of the hybrid mechanisms may be degenerate. 

In Section \ref{sec:applications} we prove hybrid admissibility for pairs of Random Serial Dictatorship and Probabilistic Serial, as well as Random Serial Dictatorship with the adaptive Boston mechanisms. 
In contrast, for the na\"ive Boston mechanism and Rank Value mechanisms we show that neither is weakly less varying than Random Serial Dictatorship. 
Furthermore, hybrids of these mechanisms will have a degree of strategyproofness of 0 (unless $\beta=0$). 
This illustrates that while hybrid admissibility is a sufficient condition, it is also close to being necessary for non-degenerate trade-offs. 
\section{Parametric Trade-offs Between Strategyproofness and Efficiency}
\label{sec:efficiency}
We have obtained a good understanding of the incentive properties of hybrid mechanisms. 
However, ultimately, we are interested in the trade-off between strategyproofness and efficiency. 
To this end, we need to understand the efficiency properties of hybrids. 
We first review three notions of dominance, namely ex-post, ordinal, and rank dominance, and the corresponding efficiency requirements. 
We then show that, loosely speaking, hybrid mechanisms inherit a share of the efficiency advantages of the more efficient component, and this share is proportional to the mixing factor $\beta$. 
\subsection{Dominance and Efficiency Notions}
\label{sec:efficiency:notions}
Ex-post efficiency is ubiquitous in assignment problems. 
Most assignment mechanisms considered in theory and applications are ex-post efficient, such as Random Serial Dictatorship, Probabilistic Serial, Rank Value mechanisms, and variants of the Boston mechanism. 
Ex-post efficiency requires that \emph{ex-post}, when every agent finally holds one object, no Pareto improvements are possible by re-assigning objects. 
\begin{definition}[Ex-post Efficient] 
\label{def:expost_eff} 
Given a preference profile $\bm P \in \mathcal{P}^N$, a deterministic assignment $x$ \emph{ex-post dominates} another deterministic assignment $y$ \emph{at} $\bm P$ if all agents weakly prefer their assigned object under $x$ to their assigned object under $y$. 
The dominance is \emph{strict} if at least one agent strictly prefers its assigned object under $x$. 
A deterministic assignment $x$ is \emph{ex-post efficient at} $\bm P$ if it is not strictly ex-post dominated by any other deterministic assignment at $\bm P$.
Finally, a random assignment is \emph{ex-post efficient at $\bm P$} if it has a lottery decomposition consisting only of deterministic assignments that are ex-post efficient at $\bm P$.
\end{definition}
To compare random assignments by their efficiency, we draw on notions of dominance for random assignments. 
\begin{definition}[Ordinally Efficient] 
\label{def:fosd_and_ordinal_eff} 
For a preference order $P : a_1 \succ \ldots \succ a_m$ and two assignment vectors $v = v_{j \in M}$ and $ w = w_{j \in M}$, we say that $v$ \emph{first order-stochastically dominates} $w$ at $P$ if for all ranks $k \in \{1,\ldots,m\}$ we have 
\begin{equation}
\sum_{j \in M: j \succ a_k} v_{j} \geq \sum_{j \in M: j \succ a_k} w_{j}.
\label{eq:fosd}
\end{equation}
For a preference profile $\bm P$, an assignment $x$ \emph{ordinally dominates} another assignment $y$ \emph{at} $\bm P$ if for all agents $i \in N$, the assignment vector $x_i$ first order-stochastically dominates $y_i$ at $P_i$. 
$x$ \emph{strictly} ordinally dominates $y$ at $\bm P$ if in addition inequality (\ref{eq:fosd}) is strict for some agent $i \in N$ and some rank $k \in \{1,\ldots,m\}$.
Finally, $x$ is \emph{ordinally efficient at} $\bm P$ if it is not strictly ordinally dominated by any other assignment at $\bm P$.
\end{definition}
If $x$ ordinally dominates $y$ at $\bm P$ and $\bm P$ is the true preference profile of the agents, then all agents will prefer $x$ to $y$, independent of their underlying utility functions. 
\citet{Bogomolnaia2001ANewSolution} showed that the Probabilistic Serial mechanism produces ordinally efficient assignments (at the reported preference profiles). 
Moreover, these assignments may strictly ordinally dominate the assignments obtained from Random Serial Dictatorship at the same preference profiles. 

\citet{Featherstone2011RankBasedRefinementWP} introduced a strict refinement of ordinal efficiency, called \emph{rank efficiency}, and he developed Rank Value mechanisms that produce rank efficient assignments.
\begin{definition}[Rank Efficient] 
\label{def:rank_dom_eff} 
For a preference profile $\bm P$ let $\text{ch}_{P_i}(k)$ denote the $k$th choice object of the agent $i$ with preference order $P_i$. 
The \emph{rank distribution} of an assignment $x$ at $\bm P$ is a vector $\bm{d}^x = (d_1^x,\ldots,d_m^x)$ with
\begin{equation}
d_k^x = \sum_{i \in M} x_{i,\text{ch}_{P_i}(k)} \text{ for }k\in \{1,\ldots,m\}.
\label{eq:def_rank_dist}
\end{equation}
$d_k^x$ is the expected number of $k$th choices assigned under $x$ with respect to preference profile $\bm P$.
An assignment $x$ \emph{rank dominates} another assignment $y$ at $\bm P$ if $\bm{d}^x$ first order-stochastically dominates $\bm{d}^y$ (i.e., $\sum_{k=1}^r d_k^x - d_k^y\geq 0 $ for all $r \in \{1,\ldots,m\}$).
$x$ \emph{strictly rank dominates} $y$ at $\bm P$ if this inequality is strict for some rank $r\in\{1,\ldots,m\}$.
$x$ is \emph{rank efficient} at $\bm P$ if is not strictly rank dominated by any other assignment at $\bm P$.
\end{definition}
Rank dominance captures the intuition that, for society as a whole, assigning two first choices and one second choice is preferable to assigning one first and two second choices. 
Rank efficient mechanisms in the assignment domain correspond to \emph{positional scoring rules} in the social choice domain \citep{XiaConitzer2008GeneralizedScoringRulesAndFrequencyOfCoalitionalManipulability} because they can be interpreted as maximizing an aggregate score based on ranks \citep{Featherstone2011RankBasedRefinementWP}. 
\subsection{Efficiency of Hybrid Mechanisms}
\label{sec:efficiency:result}
Using the notions of ex-post efficiency, ordinal dominance, and rank dominance, we show that hybrids inherit a share of the good efficiency properties from the more efficient component.
\newpage
\begin{theorem}
\label{thm:efficiency}
\ThmEfficiencyStatement
\end{theorem}
The proof is given in Appendix \ref{app:proofs:efficiency}. 

Theorem \ref{thm:efficiency} shows that hybrid mechanisms inherit a part of the desirable efficiency properties from their more efficient component. 
Statement \ref{itm:thm:efficiency:expost} is important to ensure that the baseline requirement of ex-post efficiency is preserved. 
Statement \ref{itm:thm:efficiency:dominance} yields that if the component $\g$ is more efficient than the component $\f$ in the sense of ordinal or rank dominance, then all hybrids will have intermediate efficiency (i.e., $\h$ will dominate $\f$ but be dominated by $\g$). 
Furthermore, it is straightforward to see that under these conditions, efficiency improves as $\beta$ increases: 
consider two different hybrid mechanisms $\h$ and $h^{\beta'}$ with $\beta < \beta'$. 
By setting $\beta^* = \frac{\beta'-\beta}{1-\beta}$, we can write 
\begin{equation}
	h^{\beta'} = (1-\beta) \cdot \f + \beta \cdot \g = (1-\beta^*) \cdot h^{\beta} + \beta^* \cdot \g 
\end{equation}
as a $\beta^*$-hybrid with components $h^{\beta}$ and $\g$. 
Consequently, hybrids with higher mixing factors dominate hybrids with lower mixing factors. 

However, not all mechanisms are comparable everywhere. 
For example, the Probabilistic Serial mechanism is ordinally efficient, but it does not ordinally dominate the ordinally inefficient Random Serial Dictatorship mechanism at \emph{all} preference profiles. 
Instead, some assignments resulting under the two mechanisms may not be comparable by ordinal dominance. 
In these cases, the second direction of the equivalence in statement \ref{itm:thm:efficiency:dominance} becomes useful: 
when dominance does not permit a clear decision between assignments, then the hybrid will not have clearly worse efficiency than either component. 
Thus, intuitively, efficiency of the hybrid $\h$ is better than the efficiency of $\f$ \emph{whenever this statement is meaningful}. 
\newpage
\subsection{A Parametric Measure for Efficiency Gains}
\label{sec:efficiency:measure}
\enlargethispage{1em}
Hybrid mechanisms yield a natural measure for efficiency gains, namely the mixing factor $\beta$. 
First, consider a preference profile $\bm P$ and two mechanisms $\f,\g$, such that 
$\g(\bm P)$ ordinally dominates $\f(\bm P)$ at $\bm P$. 
Independent of the particular vNM utility functions underlying the agents' ordinal preferences, we know that every agent has (weakly) higher expected utility under $\g(\bm P)$ than under $\f(\bm P)$. 
Moreover, the agents' expected utility under the hybrid $\h$ is a linear function of the mixing factor because 
\begin{equation}
	\mathbb{E}_{\h_i(\bm P)}[u_i] = (1-\beta) \cdot \mathbb{E}_{\f_i(\bm P)}[u_i] + \beta \cdot \mathbb{E}_{\g_i(\bm P)}[u_i].
\end{equation}
Thus, the gain in any agent's expected utility from using $\h$ rather than $\f$ is exactly the $\beta$-share of the gain in the agent's expected utility from using $\g$ rather than $\f$. 

Second, suppose that $\g(\bm P)$ rank dominates $\f(\bm P)$ at $\bm P$. 
A \emph{rank valuation} $v:\{1,\ldots,m\} \rightarrow \mathds{R}$ with $v(k) \geq v(k+1)$ is a function that associates a value $v(k)$ with giving some agent its $k$th choice object. 
The \emph{$v$-rank value} of an assignment $x \in \Delta(X)$ is the aggregate expected value from choosing $x$ and it is given by 
\begin{equation}
	v(x,\bm P) = \sum_{k = 1}^m d^x_k \cdot v(k).
\end{equation}
Consequently, the $v$-rank value of the hybrid $\h$ is a linear function of the mixing factor: 
\begin{equation}
	v(\h(\bm P), \bm P) = (1-\beta) \cdot v(\f(\bm P), \bm P) + \beta \cdot v(\g(\bm P), \bm P).
\end{equation}
The fact that $\g(\bm P)$ rank dominates $\f(\bm P)$ implies that the $v$-rank value of $\g(\bm P)$ is (weakly) higher than the $v$-rank value of $\f(\bm P)$ for any rank valuation $v$ \citep{Featherstone2011RankBasedRefinementWP}. 
Thus, the gain in $v$-rank value from using $\h$ rather than $\f$ is exactly the $\beta$-share of the gain in $v$-rank value from using $\g$ rather than $\f$. 

\medskip
In combination, Theorems \ref{thm:construction} and \ref{thm:efficiency} show that parametric trade-offs between strategyproofness (measured by the degree of strategyproofness $\rho$) and efficiency (measured by the mixing factor $\beta$) are possible via hybrid mechanisms: 
when a pair of mechanisms is hybrid-admissible and the second component dominates the first, a higher mixing factor will yield hybrids that are more efficient (whenever such a statement is meaningful) but also have lower degree of strategyproofness. 
\subsection{Computability of Maximal Mixing Factor}
\label{sec:efficiency:compute}
Given our understanding of hybrids, the question arises how a mechanism designer can use this construction to perform a trade-off between strategyproofness and efficiency. 
Suppose that a minimal acceptable degree of strategyproofness $\underline\rho$ is given. 
Then the mechanism designer faces the computational problem of \emph{finding the highest mixing factor $\beta$, such that $\h$ remains $\underline\rho$-partially strategyproof}.
Formally, she is interested in 
\begin{equation}
	\bmax_{\NMQ,\f,\g}(\underline\rho) = \max\left\{ \beta \in [0,1] ~|~ \h\text{ is }\underline\rho\text{-partially strategyproof in }\NMQ \right\}
\end{equation}
In \citep{MennleSeuken2017PSP_WP}, we have shown that the degree of strategyproofness $\rho_{\NMQ}(\f)$ is computable. 
Thus, we have a solution to the problem of ``finding $\rho(\h)$, given $\beta$.'' 
However, the mechanism designer's problem is the \emph{inverse} of this problem, namely to ``find $\beta$, given $\underline\rho$.'' 
The following algorithm solves this problem. 
\begin{algorithm}[ht!]
\caption{\BetaMax}
\textbf{Input}: setting \NMQ, mechanisms $\f,\g$, bound $\underline\rho$ \\
\textbf{Variables}: agent $i$, preference profile $(P_i,P_{-i})$, misreport $P_i'$, vectors $\delta^{\f},\delta^{\g}$, rank $K$, choice function $\text{ch}$, real values $\bmax, p_K^{\f}, p_K^{\g}$ \\
\Begin{
	$\bmax \gets 1$ \\
	\For{$i \in N, (P_i,P_{-i}) \in \mathcal{P}^N, P_i' \in \mathcal{P}$}{
		$\forall j \in M: \delta^{\f}_j \gets \f_{i,j}(P_i,P_{-i}) - \f_{i,j}(P_i',P_{-i})$\\
		$\forall j \in M: \delta^{\g}_j \gets \g_{i,j}(P_i,P_{-i}) - \g_{i,j}(P_i',P_{-i})$\\
		\For{$K \in \{1,\ldots,m\}$}{
			$p_K^{\f} \gets \sum_{k = 1}^K \delta^{\f}_{\text{ch}_{P_i}(k)} \cdot \underline\rho^k$\\
			$p_K^{\g} \gets \sum_{k = 1}^K \delta^{\g}_{\text{ch}_{P_i}(k)} \cdot \underline\rho^k$\\
			\If{$p_K^{\g} < 0$}{
				$\bmax \gets \min\left\{\bmax,~ p_K^{\f} /\left(p_K^{\f}-p_K^{\g}\right)\right\}$
			}
		}
	}
	\Return $\bmax$
}
\label{alg:compute}
\end{algorithm}

Algorithm \ref{alg:compute} optimistically sets its guess of $\bmax$ to $1$. 
Then it iterates through all possible preference profiles, all agents, and all misreports that agents may submit. 
For each of these combinations, it uses the partial dominance interpretation of partial strategyproofness (Theorem 4 in \citep{MennleSeuken2017PSP_WP}) to determine whether the current guess is too high, and the value is adjusted downward if necessary. 
\begin{proposition}
\label{prop:computable}
\PropComputabilityStatement
\end{proposition}
The proof is given in Appendix \ref{app:proofs:computable}. 
\subsubsection{Computational Cost of BetaMax}
\label{sec:efficiency:compute:cost}
Note that our main goal is to show \emph{computability}, not computational efficiency. 
Nonetheless, we can make a statement about the computational cost of running \BetaMax: 
computing the random assignments from the mechanisms $\f$ and $\g$ may itself be a costly operation.\footnote{Determining the probabilistic assignment of a mechanism may be computationally hard, even if implementing the mechanism is easy (e.g., see \citep{Aziz2013ComplexityRSD}).}
Thus, if $O(\f)$ and $O(\g)$ denote the cost of determining $\f$ and $\g$ for a single preferences profile, respectively, then the overall cost of Algorithm \ref{alg:compute} is $O\left(n\cdot m \cdot (m!)^{n+1} \left( O(\f) + O(\g)\right)\right)$. 

In the most general case (i.e., without any additional restrictions), a mechanism is specified in terms of a set of assignment matrices $\{\f(\bm P), \bm P \in \mathcal{P}^N\}$. 
This set will contain $(m!)^n$ matrices of dimension $n \times m$.
Consequently, the size of the problem is $S = (m!)^n \cdot n \cdot m. $
In terms of $S$, Algorithm \ref{alg:compute} has complexity $ O\left(S \sqrt[n]{S}\right)$. 
Thus, for the general case, there is not much room for improvement: 
since the algorithm must consider each preference profile at least once, \emph{any} correct and complete algorithm has computational cost of at least $S$.
\subsubsection{Reductions of Computational Cost}
\label{sec:efficiency:compute:cost_reduction}
Reductions of the computational complexity are possible if more information is available about the mechanisms $\f$ and $\g$. 
For anonymous $\f$ and $\g$, the identities of the agents is irrelevant. 
In this case, the computational cost can be reduced to 
\begin{equation}
	O\left( n \cdot m! \left(\begin{array}{c} m! + n - 1 \\ n  \end{array}\right) \left( O(\f) + O(\g)\right)\right),
\end{equation}
because only $ \left(\begin{array}{c} m! + n - 1 \\ n \end{array}\right)$ preference profiles must be considered.
Moreover, if the mechanisms are also neutral (i.e., the assignment does not depend on the objects' names either), then it suffices to consider only agent 1 with a fixed preference order. 
With this, the computational cost can be further reduced to
\begin{equation}
	O\left( m! \left(\begin{array}{c} m! + n - 2 \\ n - 1  \end{array}\right) \left( O(\f) + O(\g)\right)\right).
\end{equation}
Even with these reductions, running Algorithm \ref{alg:compute} is costly for larger settings. 
However, it is likely that more efficient algorithms exist for mechanisms with additional restrictions, and bounds may be derived analytically for certain interesting mechanisms, such as Probabilistic Serial. 
Having shown computability, we leave the design of computationally more efficient algorithms to future research.
\vspace{-.5em}
\section{Application to Pairs of Popular Mechanisms}
\label{sec:applications}
\enlargethispage{3em}
So far, we have considered abstract hybrid mechanisms and we have derived general results. 
In this section, we consider concrete instantiations of our construction. 
Indeed, it is applicable to some (but not all) well-known mechanisms. 
$\f=\text{RSD}$ is a canonical choice because it is the only known mechanism that is strategyproof, ex-post efficient, and anonymous. 
In order to apply Theorem \ref{thm:construction} (for the construction of partially strategyproof hybrids), we must establish two requirements for the second component: 
$\g$ must be upper invariant, 
and $\g$ must be weakly less varying than RSD. 
Furthermore, to obtain efficiency gains, $\g$ must be more efficient than RSD in some sense. 
Table \ref{table:overview_instantiations} provides an overview of our results. 
Trade-offs for ordinal dominance can be achieved via hybrids of RSD and PS, and trade-offs for rank dominance are possible via hybrids of RSD and ABM. 
However, NBM and RV both violate hybrid-admissibility (in combination with RSD), and we find that in fact they do not admit a non-degenerate trade-off.\footnote{We provide short descriptions of the mechanisms RSD, PS, NBM, ABM, and RV in Appendix \ref{app:mechanisms}.} 

\begin{table}[h]%
\begin{tabular}{|l|l|l||c|c|c||c|}
\hline
$\bm{\f}$ & $\bm{\g}$ & \textbf{Dominance} & \textbf{UI} & \textbf{WLV} & $\bm{(\f,\g)}$ \textbf{hyb.-adm.} & $\bm{\h}$ \textbf{PSP} \\
\hline
\hline
RSD & PS & Ordinal & \cmark & \cmark & \cmark & \cmark \\
\hline
RSD & RV & Rank & \xmark & \xmark & \xmark &\xmark \\
\hline
RSD & NBM  & Rank & \cmark & \xmark & \xmark & \xmark \\
\hline
RSD & ABM & Rank (with exceptions) & \cmark & \cmark & \cmark & \cmark \\
\hline
\end{tabular}
\caption{Results overview, \emph{UI}: $\g$ upper invariance, \emph{WLV}: $\g$ weakly less varying than $\f$, \emph{PSP}: $r$-partially strategyproof for some $r>0$.}
\label{table:overview_instantiations}
\end{table}
\subsection{Hybrids of RSD and PS}
\label{sec:applications:ps}
By Theorem 2 of \cite{Hashimoto2013TwoAxiomaticApproachestoPSMechanisms}, PS is upper invariant.
Since PS is ordinally efficient, it is never ordinally dominated by RSD at any preference profile. 
Furthermore, PS may (but does not always) ordinally dominate RSD \citep{Bogomolnaia2001ANewSolution}.
Thus, PS ordinally dominates RSD whenever the two mechanisms are comparable.
To obtain hybrid-admissibility of the pair $(\text{RSD},\text{PS})$, it remains to be shown that PS is weakly less varying than RSD.
\begin{theorem} 
\label{thm:PS_wlv_RSD} 
\ThmPSWLVStatement
\end{theorem}
\begin{proof}[Proof Outline (formal proof in Appendix \ref{app:proofs:pswlv})] 
Consider an agent $i$ that swaps two objects in its report from $P_i:a\succ b$ to $P_i':b\succ a$. 
First, we show that PS changes the assignment \emph{if and only if} neither $a$ nor $b$ are exhausted when $i$ finishes consuming objects that it strictly prefers to both. 
Next, we show that RSD changes the assignment \emph{if and only if} there exists an ordering of the agents such that all objects that $i$ prefers strictly to $a$ are assigned before $i$ gets to pick, but neither $a$ nor $b$ are assigned by then. 
Finally, we show that the first condition (for PS) implies the second condition (for RSD), using an inductive argument. 
The key idea is to show that if no such ordering of the agents exists for $m$ objects, then we can construct a case with $m-1$ objects where no such ordering exists either.
\end{proof}

\begin{corollary} 
\label{cor:rsd_ps_hybrid_admissible} 
The pair (RSD,PS) admits the construction of partially strategyproof hybrids that improve efficiency in terms of ordinal dominance.
\end{corollary}
\subsection{Two Impossibility Results}
\label{sec:applications:impossible}
A mechanism designer may also want to trade off strategyproofness for improvements of the rank distribution.
Mechanisms that aim at achieving a good rank distribution are Rank Value mechanisms \citep{Featherstone2011RankBasedRefinementWP} and Boston mechanisms \citep{MennleSeuken2017ABM_WP}.
It turns out, however, that neither RV nor the na\"{i}ve variant of the Boston mechanism (NBM) are suitable second components in combination with RSD. 
\subsubsection{Impossibility Result for Rank Value Mechanisms}
\label{sec:applications:impossible:rv}
\enlargethispage{-1em}
RV rank dominates RSD whenever their outcomes are comparable, since RV is rank efficient, but RSD is not.
However, no rank efficient mechanism can be upper invariant, as we demonstrate in Example \ref{ex:RV_mix_manip}. 
Therefore, the pair $(\text{RSD},\text{RV})$ violates hybrid-admissibility.\footnote{In addition to violating upper invariance, RV is not weakly less varying than RSD in general; see Example \ref{ex:RV_not_WLV} in Appendix \ref{app:examples}.} 
\begin{example}
\label{ex:RV_mix_manip}
Consider a setting with agents $N=\{1,2,3\}$ and objects $M = \{a,b,c\}$, each available in unit capacity. 
If the agents have preferences 
\begin{eqnarray*}
P_1,P_2 & : & a \succ b \succ c, \\
P_3 & : & c \succ a \succ b,
\end{eqnarray*}
then any rank efficient assignment must assign $c$ to agent $3$. 
Therefore, at least one of the agents $1$ and $2$ has a positive probability for $b$. 
Without loss of generality, let $\text{RV}_{1,b}(\bm P) > 0$. 
If agent 1 reports 
\begin{eqnarray*}
P_1' & : & a \succ c \succ b
\end{eqnarray*}
instead, then the unique rank efficient assignment must assign $a$ to agent 1. 
Since this misreport changes agent 1's assignment of $a$, RV is not upper invariant. 
\end{example}
It follows from Example \ref{ex:RV_mix_manip} that any non-trivial hybrid $\h$ of RSD and RV will violate upper invariance. 
This means that $\h$ will not be $r$-partially strategyproof for any positive $r > 0$ (by Fact \ref{fact:thm2}), or equivalently, $\h$ will have a degree of strategyproofness of $0$. 
This teaches us that RSD and RV indeed do not admit the construction of hybrid mechanisms that make a non-degenerate trade-off between strategyproofness and efficiency. 
\subsubsection{Impossibility Result for the Na\"ive Boston Mechanism}
\label{sec:applications:impossible:nbm}
We consider the Boston mechanism with no priorities and single uniform tie-breaking \citep{Miralles2008CaseForBostonWP}. 
The ``na\"{i}ve'' variant of the Boston mechanism (\emph{NBM}) lets agents apply to their respective next best choices in consecutive rounds, even if the objects to which they apply have no more remaining capacity. 
The assignments from NBM rank dominate those from RSD whenever they are comparable, and NBM is also upper invariant \citep{MennleSeuken2017ABM_WP}. 
However, NBM is \emph{not} weakly less varying than RSD, as Example \ref{ex:NBM_not_WLV} shows. 
Thus, pairs of RSD and NBM violate hybrid-admissibility. 
\begin{example}
\label{ex:NBM_not_WLV} 
\ExampleNBMnotWLV
\end{example}
In fact, Example \ref{ex:NBM_not_WLV} shows something more, namely that at the particular preference profile, NBM is manipulable in a first order-stochastic dominance sense, while RSD does not change the assignment at all. 
Thus, any hybrid of RSD and NBM will also be manipulable in a first order-stochastic dominance sense at this preference profile. 
Consequently, the hybrid cannot be $r$-partially strategyproof for any $r>0$ (by Proposition 2 in \citep{MennleSeuken2017PSP_WP}). 
Analogous to the pair $(\text{RSD},\text{RV})$, we learn that the pair $(\text{RSD},\text{NBM})$ does not admit the construction of hybrid mechanisms with non-degenerate degrees of strategyproofness either. 
\subsection{Hybrids of RSD and ABM}
\label{sec:applications:abm}
In \citep{MennleSeuken2017ABM_WP}, we have formalized an adaptive variant of the Boston mechanism (ABM), which is sensitive to the fact that agents cannot benefit from applying to objects that were already exhausted in previous rounds. 
Instead, in each round, agents who have not been assigned so far, apply to their most preferred \emph{available} object.

Our analysis of ABM has revealed two further attributes: 
first, ABM is upper invariant, one of the conditions we need for hybrid-admissibility \citep{MennleSeuken2017PSP_WP}. 
Second, ABM rank dominates RSD whenever the two mechanisms are comparable, except in certain special cases \citep{MennleSeuken2017ABM_WP}. 
These exceptions occur rarely, and the probability of encountering them vanishes as markets get large.
Thus, if we can show that ABM is also weakly less varying than RSD, then we can use this pair to construct partially strategyproof hybrids that trade off strategyproofness and efficiency in terms of rank dominance (with the exception of a tiny number of preference profiles). 
\begin{theorem} 
\label{thm:ABM_wlv_RSD} 
\ThmABMWLVStatement
\end{theorem}
\begin{proof}[Proof outline (formal proof in\ Appendix \ref{app:proofs:abmwlv})] 
Both RSD and ABM are implemented by randomizing over orderings $\pi$ of agents. 
Suppose $i$ manipulates by swapping $a$ and $b$. 
If ABM changes the assignment, then there exists $\pi$ such that all objects that $i$ prefers strictly to $a$ and $b$ are assigned in previous rounds to other agents. 
Then $i$ gets to ``pick'' between $a$ and $b$ under ABM. 
Starting with $\pi$, we construct an ordering $\pi'$ such that if the ordering $\pi'$ is drawn under RSD, $i$ gets to pick between $a$ and $b$ but no object it strictly prefers to $a$ or $b$ under RSD. 
This is sufficient for RSD to also change the assignment. 
\end{proof}

\begin{corollary} 
\label{cor:rsd_abm_hybrid_admissible} 
The pair (RSD,ABM) admits the construction of partially strategyproof hybrids that improve efficiency in terms of rank dominate (with few exceptions). 
\end{corollary}
\subsection{Numerical Results}
\label{sec:applications:numeric}
We have shown that we can construct interesting hybrids by combining RSD with PS or ABM.
This gives mechanism designers the possibility to trade off strategyproofness for better efficiency.
To illustrate the magnitude of these trade-offs, we have computed $\bmax$ for a variety of settings \NMQ\ and acceptable degrees of strategyproofness $\rho \in [0,1]$. 

\begin{figure}[t]
\begin{center}
\includegraphics[width=\columnwidth]{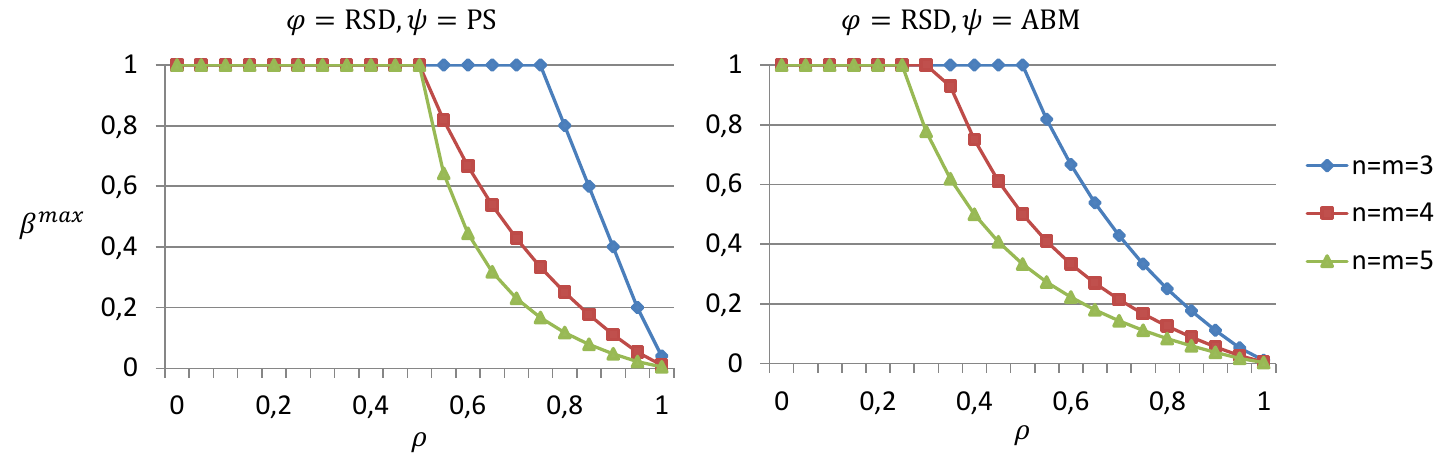}
\caption{Plots of $\bmax$ by acceptable degree of strategyproofness $\rho \in [0,1]$, for components $\f=\text{RSD}, \g\in \{\text{PS},\text{ABM}\}$, and settings with $n = m \in \{3,4,5\}$.}
\label{fig:beta_max_PS_ABM}
\end{center}
\end{figure}
Figure \ref{fig:beta_max_PS_ABM} shows plots of the maximal mixing factor $\bmax$ for settings with unit capacity and different numbers of objects and agents.
Observe that as the acceptable degree of strategyproofness for the hybrid increases, the allowable share of $\g$ decreases and becomes $0$ if full strategyproofness is required.
We also see that the relationship between $\rho$ and $\bmax$ is not linear.
In particular, the first efficiency improvements (from $\bmax = 0$ to $\bmax > 0$) are the most ``costly'' in terms of a reduction of the degree of strategyproofness $\rho$.
On the other hand, for mild strategyproofness requirements, the share of PS or ABM in the hybrid can be significant, e.g., more than $30\%$ of PS or $17\%$ of ABM for $\rho = 0.75$ and $n=m=4$.

\begin{figure}[t]
\begin{center}
\includegraphics[width=0.58\columnwidth]{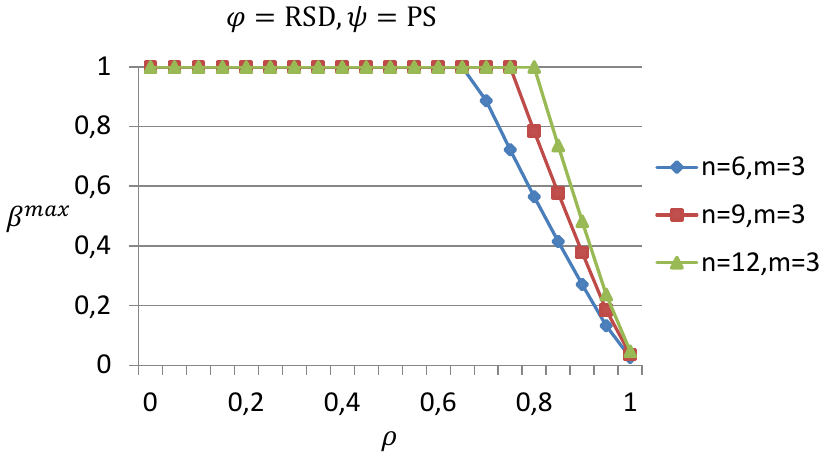}
\caption{Plots of $\bmax$ by acceptable degree of strategyproofness $\rho \in [0,1]$, for components $\f=\text{RSD}, \g=\text{PS}$, and settings with $m = 3, n \in \{6,9,12\}$.}
\label{fig:beta_max_PS_more_capacity}
\end{center}
\end{figure}
Figure \ref{fig:beta_max_PS_more_capacity} shows plots of $\bmax$ for hybrids of RSD and PS, where we hold the number of objects constant at $m = 3$ but vary the capacity of the objects $q\in \{2,3,4\}$ (with $n=q\cdot m$ agents).
We observe that for larger capacities, the hybrids can contain a larger share of PS.
This is consistent with findings by \citet{Kojima2010IncentivesInTheProbabilisticSerial}, who have shown that for a fixed agent and a fixed number of objects, PS makes truthful reporting a dominant strategy if the capacities of the objects are sufficiently high. 
It is conceivable that the degree of strategyproofness of PS keeps increasing and converges to 1 in the limit as capacity increases, an interesting question for future research.
\section{Conclusion}
\label{sec:conclusion}
In this paper, we have presented a novel approach to trading off strategyproofness and efficiency for random assignment mechanisms. 
We have introduced hybrid mechanisms, which are convex combinations of two component mechanisms, as a method to facilitate these trade-offs. 
Typically, the first component $\f$ introduces better incentives while the second component $\g$ introduces better efficiency.

For our first result, we have employed partial strategyproofness, a new concept for relaxing strategyproofness in a parametric way that we have introduced in \citep{MennleSeuken2017PSP_WP}. 
If (1) $\f$ is strategyproof and (2) $\g$ is upper invariant and (3) weakly less varying than $\f$, we have shown that partially strategyproof hybrids can be constructed for any desired degree of strategyproofness. 
At the same time, our hybrid-admissibility requirement is tight in the sense that none of the three conditions can be dropped without risking degenerate trade-offs.

For our second result, we have shown that hybrids inherit ex-post efficiency from their components, and their efficiency (relative to the components) can be understood in terms of ordinal (or rank) dominance. 
This means that, in line with intuition, hybrid mechanisms in fact trade off strategyproofness for efficiency: 
as the mixing factor $\beta$ (i.e., the share of $\g$) increases, efficiency of the hybrid increases, but the degree of strategyproofness decreases. 
This has important consequences for mechanism designers: 
if $\f$ is a strategyproof mechanism, $\g$ is a non-strategyproof alternative that is more appealing due to its efficiency properties, and a certain degree of strategyproofness $\underline\rho < 1$ is acceptable, then a hybrid can be used to improve efficiency, subject to the $\underline\rho$-partial strategyproofness constraint. 
As we have shown in Section \ref{sec:efficiency:compute}, the mechanism designer's problem of determining the maximal mixing factor can be solved algorithmically.

Finally, we have presented instantiations of hybrid mechanisms with $\f=\text{RSD}$ as the strategyproof component. 
Using $\g = \text{PS}$ yields better efficiency in an ordinal dominance sense, and using $\g = \text{ABM}$, an adaptive variant of the Boston mechanism, yields better efficiency in a rank dominance sense (with few exceptions). 
Numerically, we have illustrated the connection between the degree of strategyproofness $\rho$ and the maximal mixing factor $\bmax$, and we have shown that the latter can be significant for even mild reductions of the minimal acceptable degree of strategyproofness.

This paper contributes to an important area of research concerned with trade-offs between strategyproofness and efficiency in the assignment domain. 
Hybrid mechanisms break new ground because the method is constructive, it enables a parametric trade-off, and the mechanism designer's problem of determining a suitable hybrid is computable. 
Our hybrids shed light on the frontiers of such trade-offs and can serve as benchmark mechanisms for future research.
\begin{small}


\end{small}

\newpage
\appendix
\section*{APPENDIX}
\setcounter{section}{0}
\section{Mechanisms}
\label{app:mechanisms}
We explain how each mechanism determines the assignment based on a reported profile $\bm P$ of preferences. 
\subsection{Random Serial Dictatorship Mechanism}
The \emph{Random Serial Dictatorship} mechanism chooses an agent uniformly at random and assigns this agent its first choice object. 
Next, it chooses another agent uniformly at random from the remaining agent and assigns this agent the object that it prefers most out of all the objects that have remaining capacity. 
This continues until all agents have received an object. 
The random assignment matrix arises from the fact that agents do not know when they will be chosen by the mechanism. 
\subsection{Probabilistic Serial Mechanism}
Under the \emph{Probabilistic Serial} mechanism, the objects are treated as if they were divisible. 
All agents start consuming probability shares of their first choice objects at equal speeds. 
Once all capacity of an object is completely consumed, all agents who were consuming this object, move on to their next preferred object. 
If this next object is already exhausted as well, they go directly to the next object, and so on. 
This process continues until all agents have collected a total of 1 units of some objects. 
The shares of objects that each agent has collected are the entries in the assignment matrix of the Probabilistic Serial mechanism. 
\subsection{Na\"ive Boston Mechanism}
Under the \emph{na\"ive Boston} mechanism, all agents report their preferences and then draw a random number. 
The assignment process occurs in rounds. 
In the first round, each agent applies to its most preferred object. 
Applicants are assigned the objects to which they applied if these have sufficient capacity. 
If an object has more applicants than remaining capacity, preference is give to agents with higher random numbers. 
The agents who did not get an object in the first round continue to the second round. 
In the $k$th round, each remaining agent applies to its $k$th choice. 
Again, objects are assigned to agents until their capacity is exhausted, and the unlucky agents with the lowest random numbers enter the next round. 
The assignment process ends when all agents have received an object. 
The random assignment matrix arises from the fact that agents do not know their random numbers. 
\subsection{Adaptive Boston Mechanism}
The \emph{adaptive Boston} mechanism works like the na\"ive Boston mechanism, except that in each round, the remaining agents apply to the object that they prefer most out of all the objects that still have remaining capacity. 
Again, the random assignment matrix arises from the fact that agents do not know their random numbers. 
\subsection{Rank Value Mechanism}
\emph{Rank Value} mechanisms are a class of mechanisms. 
Given a rank valuation $v : \{1,\ldots,m\} \rightarrow \mathds{R}$ with $v(k) \geq v(k+1)$, a $v$-Rank Value mechanism determines an assignment by solving the following linear program:
\begin{eqnarray*}
	\text{maximize} & & \sum_{i\in N} \sum_{j \in M} v(\text{rank}_{P_i}(j)) \cdot x_{i,j}, \\
	\text{subject to} & & \sum_{i\in N} x_{i,j} = 1, \text{ for all }j \in M, \\
	& & \sum_{j \in M} x_{i,j} \leq q_j, \text{ for all }i \in N, \\
	& & x_{i,j} \in [0,1], \text{ for all }i\in N, j \in M,
\end{eqnarray*}
where $\text{rank}_{P_i}(j)$ is the rank of $j$ under the preference ranking of agent $i$, i.e., the number of objects that this agent weakly prefers to $j$. 
\section{Example from Section \ref{sec:applications:impossible}}
\label{app:examples}
\begin{example}[RV not Weakly Less Varying than RSD] 
\label{ex:RV_not_WLV} 
\ExampleRVnotWLV
\end{example}
\section{Omitted Proofs} 
\label{app:proofs}
\subsection{Proof of Theorem \ref{thm:construction}}
\label{app:proofs:construction}
\begin{proof}[Proof of Theorem \ref{thm:construction}]
\textit{\ThmConstructionStatement} 

\medskip
\ThmConstructionProof
\end{proof}
\subsection{Proof of Theorem \ref{thm:efficiency}}
\label{app:proofs:efficiency}
\begin{proof}[Proof of Theorem \ref{thm:efficiency}]
\textit{\ThmEfficiencyStatement} 

\medskip
\ThmEfficiencyProof
\end{proof}
\subsection{Proof of Proposition \ref{prop:computable}}
\label{app:proofs:computable}
\begin{proof}[Proof of Proposition \ref{prop:computable}]
\textit{\PropComputabilityStatement} 

\medskip
\PropComputabilityProof
\end{proof}
\subsection{Proof of Theorem \ref{thm:PS_wlv_RSD}}
\label{app:proofs:pswlv}
\begin{proof}[Proof of Theorem \ref{thm:PS_wlv_RSD}]
\textit{\ThmPSWLVStatement} 

\medskip
\ThmPSWLVProof
\end{proof}
\subsection{Proof of Theorem \ref{thm:ABM_wlv_RSD}}
\label{app:proofs:abmwlv}
\begin{proof}[Proof of Theorem \ref{thm:ABM_wlv_RSD}]
\textit{\ThmABMWLVStatement} 

\medskip
\ThmABMWLVProof
\end{proof}
\end{document}